\documentclass[twoside,leqno,twocolumn]{article}
\usepackage{ltexpprt}

\usepackage[bookmarks,bookmarksopen,bookmarksdepth=2]{hyperref}
\usepackage{xspace}
\usepackage{comment}
\usepackage{amsmath}
\usepackage{amsfonts}
\usepackage{amssymb}
\usepackage{pdfsync}
\usepackage{thmtools,thm-restate}
\usepackage[textsize=small]{todonotes}
\usepackage{paralist}
\usepackage[capitalize]{cleveref}
\usepackage{authblk}

\usepackage{subcaption}

\usepackage{tikz}
\usetikzlibrary{arrows,shapes,snakes,automata, calc, chains, matrix, positioning, scopes, fit, backgrounds}

\newtheorem{remark}{Remark}[section]
\newtheorem{claim}{Claim}[section]

\newcommand{\ignore}[1]{}

\newcommand{\C}{\mathcal{C}}

\newcommand{\Rr}{\mathcal{R}}

\newcommand{\A}{\mathcal{A}}

\renewcommand{\S}{\mathcal{S}}
\newcommand{\T}{\mathcal{T}}
\newcommand{\U}{\mathcal{U}}

\newcommand{\Pp}{\mathcal{P}}

\newcommand{\set}[1]{\{#1\}}
\newcommand{\eset}[1]{\{\,#1\,\}}
\newcommand{\seq}[1]{\langle #1 \rangle}

\newcommand{\tree}{\textup{tree}}

\newcommand{\evenloops}{\textup{EvenLoops}\xspace}
\newcommand{\oddloops}{\textup{OddLoops}\xspace}
\newcommand{\limsupeven}{\textup{LimsupEven}\xspace}
\newcommand{\limsupodd}{\textup{LimsupOdd}\xspace}
\newcommand{\evencycles}{\textup{EvenCycles}\xspace}
\newcommand{\oddcycles}{\textup{OddCycles}\xspace}

\newcommand{\eps}{\varepsilon}

\newcommand{\floor}[1]{\lfloor #1 \rfloor}
\newcommand{\ceil}[1]{\lceil #1 \rceil}

\begin{document}

\title{Universal trees grow inside separating automata: \\
  Quasi-polynomial lower bounds for parity games}

\author[1]{Wojciech Czerwi\'nski}
\author[2]{Laure Daviaud}
\author[3,4,2]{Nathana\"el Fijalkow}
\author[2]{Marcin~Jurdzi\'nski}
\author[2]{Ranko~Lazi\'c}
\author[1]{Pawe{\l} Parys}
\affil[1]{MIMUW, University of Warsaw, Poland}
\affil[2]{DIMAP, Department of Computer Science, University of Warwick, UK}
\affil[3]{CNRS, LaBRI, France}
\affil[4]{The Alan Turing Institute, UK}
\date{}                     

\maketitle

\begin{abstract}
  Several distinct techniques have been proposed to design
  quasi-polynomial algorithms for solving parity games since the
  breakthrough result of Calude, Jain, Khoussainov, Li, and
  Stephan~(2017):    
  \emph{play summaries}, \emph{progress measures}
  and \emph{register games}. 
  We argue that all those techniques can be viewed as instances of the
  \emph{separation approach} to solving parity games, a key technical
  component of which is constructing (explicitly or implicitly) an
  automaton that \emph{separates} languages of words encoding plays
  that are (decisively) won by either of the two players.    
  Our main technical result is a quasi-polynomial lower bound on the
  size of such separating automata that nearly matches the current
  best upper bounds.
  This forms a \emph{barrier} that all existing approaches must
  overcome in the ongoing quest for a polynomial-time algorithm for
  solving parity games. 
  The key and fundamental concept that we introduce and study is a
  \emph{universal ordered tree}.
  The technical highlights are a quasi-polynomial lower bound on the
  size of universal ordered trees and a proof that every 
  \emph{separating safety automaton} has a universal tree
  hidden in its state space.
\end{abstract}

\section{Introduction}
\label{sec:intro}

\subsection{Parity games.}
A \emph{parity game} is played on a directed graph by two players who
are called Even and Odd.   
A play starts at a designated vertex and then the players move by
following outgoing edges forever, thus forming an infinite path. 
Every vertex of the graph is owned by one of the two players and it is
always the owner of the vertex who moves by following an outgoing
edge from the current vertex to the next one.  

This completes the description of the dynamics of a play, but how do
we declare the winner of an infinite path formed in this way?
For this, we need to inspect positive integers that label all edges in
the graph, which we refer to as \emph{edge priorities}, or simply
priorities. 
Player Even is declared the winner of a play if the highest priority
that occurs infinitely many times is even, and otherwise player Odd
wins; 
equivalently, the winner is the parity of the limsup (limes superior) 
of the priorities that occur in the play. 

The principal algorithmic problem studied in the context of parity
games is \emph{deciding the winner}:
given a game graph as described above and a starting vertex, does
player Even have a winning strategy---a recipe for winning every play
starting from the designated vertex, no matter what edges her opponent 
Odd follows whenever it is his turn to move.

\subsubsection*{Determinacy and complexity.} 
A \emph{positional strategy} for Even is a set of edges that go out of 
vertices she owns---exactly one such edge for each of her vertices; 
Even uses such a strategy by always---if the current vertex is owned
by her---following the unique outgoing edge that is in the strategy.  
Note that when Even uses a positional strategy, her moves depend only
on the current vertex---they are oblivious to what choices were made
by the players so far. 
A basic result for parity games that has notable implications is their  
\emph{positional determinacy}~\cite{EJ91,Mos91}:
for every starting vertex, exactly one of the players has a winning
strategy and hence the set of vertices is partitioned into the
\emph{winning set} for Even and the winning set for Odd;
moreover, each player has a \emph{positional strategy} that
is winning for her from all starting vertices in her winning set. 

An important corollary of positional determinacy is that deciding the
winner in parity games is \emph{well characterized}, i.e., it is both
in NP and in co-NP~\cite{EJS93}. 
Several further complexity results suggest that it may be difficult to
provide compelling evidence for hardness of solving parity games: 
deciding the winner is known to be also in UP and in
co-UP~\cite{Jur98}, and computing winning strategies is in PLS, PPAD,
and even in their subclass CLS~\cite{DP11,DTZ18}. 
Parity games share this intriguing complexity-theoretic status with
several other related problems, such as mean-payoff games~\cite{ZP96},
discounted games, and simple stochastic games~\cite{Con92}, but they
are no harder than them since there are polynomial reductions from
parity games to mean-payoff games, to discounted games, and to simple
stochastic games~\cite{Jur98,ZP96}. 

\subsubsection*{Significance and impact.}
Parity games play a fundamental role in automata theory, logic, and
their applications to verification and synthesis.
Specifically, the algorithmic problem of deciding the winner in parity
games is polynomial-time equivalent to the model checking in the modal 
$\mu$-calculus and to checking emptiness of automata on infinite trees
with parity acceptance conditions~\cite{EJS93}, and it is at the heart
of algorithmic solutions to the Church's synthesis
problem~\cite{Rab72}. 

The impact of parity games goes well beyond their place of origin in
automata theory and logic. 
We illustrate it by the resolutions of two long-standing open problems 
in \emph{stochastic planning} and in \emph{linear programming},
respectively, that were directly enabled by the ingenious examples of 
parity games given by Friedmann~\cite{Fri09}, on which the 
\emph{strategy improvement algorithm}~\cite{VJ00} requires
exponentially many iterations. 
Firstly, Fearnley~\cite{Fea10} has shown that Friedmann's examples can
be adapted to prove that \emph{Howard's policy iteration} algorithm
for Markov decision processes (MDPs) requires exponentially many
iterations.   
Policy iteration has been well-known and widely used in stochastic
planning and AI since 1960's, and it has been celebrated for its 
fast termination: until Fearnley's surprise result, no examples were
known for which a superlinear number of iterations was necessary.
Secondly, Friedmann, Hansen, and Zwick~\cite{FHZ11} have adapted the
insights from the lower bounds for parity games and MDPs to prove that
natural \emph{randomized pivoting rules} in the 
\emph{simplex algorithm} for linear programming may require
subexponentially many iterations.   
The following quote from the full version of Friedmann et
al.~\cite{FHZ11} 
highlights the role that parity games (PGs) played in their
breakthrough: 
\begin{quote}
  ``our construction can be described and understood without knowing
  about PGs.  
  We would like to stress, however, that most of our intuition about
  the problem was obtained by thinking in terms of PGs.  
  Thinking in terms of MDPs seems harder, and we doubt whether we
  could have obtained our results by thinking directly in terms of
  linear programs.''
\end{quote}
In both cases, Friedmann's examples of parity games and their analysis
have been pivotal in resolving the theoretical worst-case complexity
of influential algorithms that for many decades resisted rigorous 
analysis while performing outstandingly well in practice. 

\subsubsection*{Current state-of-the-art.}
It is a long-standing open question whether there is a polynomial-time
algorithm for solving parity games~\cite{EJS93}.
The study of algorithms for solving parity games has been dominated
for over two decades by algorithms whose run-time was exponential in
the number of distinct
priorities~\cite{EL86,BCJLM97,Sei96,Jur00,VJ00,Sch17}, or mildly 
subexponential for large number of priorities~\cite{BV07,JPZ08}.
The breakthrough came in 2017 from Calude et al.~\cite{CJKLS17} 
who gave the first quasi-polynomial-time algorithm using the novel
idea of \emph{play summaries}.
Several other quasi-polynomial-time algorithms were developed soon
after, including space-efficient \emph{progress-measure} based
algorithms of Jurdzi\'nski and Lazi\'c~\cite{JL17} and of Fearnley,
Jain, Schewe, Stephan, and Wojtczak~\cite{FJSSW17}, and the algorithm
of Lehtinen~\cite{Leh18}, based on her concept of 
\emph{register games}.

\subsection{The separation approach.}

Boja\'nczyk and Czerwi\'nski~\cite[Section 3]{BC18} have observed that
the main technical contribution of Calude et al.~\cite{CJKLS17} can be
elegantly phrased using concepts from automata theory. 
They have pointed out that in order to reduce solving a parity game of
size at most~$n$ to solving a conceptually and algorithmically much
simpler \emph{safety game}, it suffices to provide a finite
\emph{safety automaton} that achieves the task of \emph{separating}
two sets $\evenloops_n$ and $\oddloops_n$ of infinite words that
describe plays on graphs of size at most~$n$ that are
\emph{decisively won} by the respective two players. 
For encoding plays in parity games, they use words in which every
letter is a pair that consists of a vertex and a priority. 
The definition of such a word being decisively won by a player 
that was proposed by Boja\'nczyk and Czerwi\'nski is that the biggest
priority that occurs on every cycle---an infix in which the first
vertex and the vertex immediately following the infix coincide---is of
her parity. 
Concerning separation, for two disjoint languages~$J$ and~$K$, we say
that a language~$S$ \emph{separates} $J$ from~$K$ if $J \subseteq S$ 
and $S \cap K = \emptyset$, and we say that an automaton~$\A$ is a 
\emph{separator} of two languages if the language $L(\A)$ of words
recognized by~$\A$ separates them. 
The main technical contribution of Calude et al.~\cite{CJKLS17} can
then be stated as constructing separators---of quasi-polynomial
size---of the languages $\evenloops_n$ and $\oddloops_n$.

Note that a separator of $\evenloops_n$ and $\oddloops_n$ has a
significantly easier task than a \emph{recognizer} of exactly the set
$\limsupeven$ of words that are won by Even---that is required to
accept all words in $\limsupeven$, and to reject all words in
$\limsupodd$, the set of all words that are won by Odd.   
Instead, a separator may reject some words won by Even and accept some
words won by Odd, as long as it accepts all words that are decisively
won by Even, and it rejects all words that are decisively won by Odd.   

What Calude et al.~\cite{CJKLS17} exploit is that if one of the
players uses a \emph{positional} winning strategy then all plays are
indeed encoded by words that are won decisively by her, no matter how
the opponent responds.  
The formalization of Boja\'nczyk and Czerwi\'nski~\cite{BC18} is
that---using positional determinacy of parity
games~\cite{EJ91,Mos91}---in order to solve a parity game of size at
most~$n$, it suffices to solve a strategically and algorithmically
much simpler \emph{safety game} that is obtained as a simple
chained product of the parity game and a safety automaton that is 
a separator of $\evenloops_n$ and~$\oddloops_n$.  

\subsection{Our contribution.}
Our main conceptual contributions include making explicit the notion
of a \emph{universal ordered tree} and unifying all the existing 
quasi-polynomial algorithms for parity
games~\cite{CJKLS17,JL17,GIJ17,FJSSW17,Leh18} as instances of the 
\emph{separation approach} proposed by Boja\'nczyk and
Czerwi\'nski~\cite{BC18}.

We point out that it is exactly the universality property of an
ordered tree that makes it suitable for serving as a witness search
space in \emph{progress measure lifting}
algorithms~\cite{Jur00,BJW02,Sch17,JL17,DJL18}, and that the running
time of such algorithms is dictated by the size of 
the universal tree used.
In particular, by proving a quasi-polynomial lower bound on the size
of universal trees in Section~\ref{sec:lower-universal}, we rule out
the hope for improving progress measure lifting algorithms to work in
sub-quasi-polynomial time by finding smaller universal trees. 
As our other main technical results in Section~\ref{sec:proof} show,
however, universal trees are fundamental not only for progress measure
lifting algorithms, but for all algorithms that follow the separation
approach.   

We argue that in the separation approach, it is appropriate to
slightly adjust the choice of languages to be separated, from
$\evenloops_n$ and $\oddloops_n$ proposed by Boja\'nczyk and
Czerwi\'nski~\cite{BC18} to the more suitable $\evencycles_n$ and
$\oddcycles_n$ 
(see Section~\ref{sec:langs} for the definitions).   
We also verify, in Section~\ref{sec:seps-everywhere}, that all the
three distinct techniques of solving parity games in quasi-polynomial
time considered in the recent literature 
(\emph{play summaries}~\cite{CJKLS17,GIJ17,FJSSW17},
\emph{progress measures}~\cite{JL17}, and
\emph{register games}~\cite{Leh18}) 
yield separators for languages $\evencycles_n$ and $\limsupodd$, which  
(as we argue in Section~\ref{sec:sep-approach}) 
makes them suitable for the separation approach.

The main technical contribution of the paper, described in
Sections~\ref{sec:lower-universal} and~\ref{sec:proof} is a proof that
every (non-deterministic) safety automaton that separates
$\evencycles_n$ from $\limsupodd$ has a number of states that is at
least quasi-polynomial. 
Recall that in Section~\ref{sec:lower-universal} we establish a
quasi-polynomial lower bound on the size of universal trees. 
Then, in Section~\ref{sec:proof}, our argument is based on proving
that in every separating automaton as above, one can define a sequence
of \emph{linear quasi-orders} on the set of states, in which each 
quasi-order is a refinement of the quasi-order that follows it in the
sequence.  
Such a sequence of linear quasi-orders can be naturally interpreted
as an \emph{ordered tree} in which every leaf is populated by at least 
one state of the automaton.   
We then also prove that the ordered tree must contain a  
\emph{universal ordered tree}, and the main result follows
from the earlier quasi-polynomial lower bound for universal trees.

Another technical highlight, presented in
Section~\ref{section:separator-from-universal-trees}, is a
construction of a separator from an arbitrary universal tree, which
together with the lower bound in Section~\ref{sec:proof} implies that
the sizes of smallest universal trees and of smallest separators
coincide.   
The correctness of the construction relies on existence of 
\emph{progress measures} that map from vertices of game graphs into
leaves of universal trees, and that witness winning strategies.  

The significance of our main technical results is that they provide
evidence against the hope that any of the existing technical
approaches to developing quasi-polynomial algorithms for solving 
parity games~\cite{CJKLS17,JL17,FJSSW17,Leh18} may lead to further
improvements to sub-quasi-polynomial algorithms. 
In other words, our quasi-polynomial lower bounds for universal trees
and separators form a \emph{barrier} that all existing approaches must
overcome in the ongoing quest for a polynomial-time algorithm for
solving parity games.

\section{Progress measures and universal trees}

Progress measures are witnesses of winning strategies 
that underpin the design of the
\emph{small progress measure lifting algorithm}~\cite{Jur00}. 
For nearly two decades, this algorithm and its
variants~\cite{BJW02,Sch17,CHL17,JL17,DJL18} have been
consistently matching or beating the worst-case performance guarantees
of the state-of-the-art algorithms for solving parity games. 

In this section we introduce the notion of
\emph{universal ordered trees} and we point out how the universality 
property uniformly justifies the correctness of progress measure
lifting algorithms, and that the size of universal trees used in such
algorithms drives their worst-case run-time complexity. 
Those observations motivate the key technical question that we tackle
in this section, namely whether the recent construction of
quasi-polynomial universal trees by Jurdzi\'nski and
Lazi\'c~\cite{JL17} can be significantly improved to yield a 
sub-quasi-polynomial algorithms for solving parity games.

The main original technical contribution of this section is a negative
answer to this question: 
we establish a quasi-polynomial lower bound on the size of universal
trees that nearly matches (up to a small polynomial factor) the upper
bound of Jurdzi\'nski and Lazi\'c.
This dashes the hope of obtaining a significantly faster algorithm for
solving parity games by constructing smaller universal trees and
running the progress measure algorithm on them.

\subsection{Game graphs and strategy subgraphs.}
Throughout the paper, we write~$V$ for the set of vertices and $E$ for
the set of edges in a parity game graph, and we use $n$ 
to denote the numbers of vertices. 
For every edge $e \in E$, its \emph{priority} $\pi(e)$ is a positive
integer, and we use~$d$ to denote the smallest even number that
priority of no edge exceeds. 
Without loss of generality, we assume that every vertex has at least
one outgoing edge. 
We say that a cycle in a game graph is \emph{even} if the largest
edge priority that occurs on it is even;
otherwise it is \emph{odd}. 

Recall that a positional strategy for Even is a set of edges that go
out of vertices she owns---exactly one such edge for each of her
vertices.    
The \emph{strategy subgraph} of a positional strategy for Even is the
subgraph of the game graph that includes all outgoing edges from
vertices owned by Odd and exactly those outgoing edges from vertices
owned by Even that are in the positional strategy.
Observe that the set of plays that arise from Even playing her
positional strategy is exactly the set of all plays in the strategy
subgraph. 
Moreover, note that every cycle in the strategy subgraph of a
positional strategy for Even that is winning for her is even: 
otherwise, by repeating an odd cycle indefinitely, we would get a play
that is winning for Odd.
Further overloading terminology, we say that a (parity game) graph is 
\emph{even} if all cycles in it are even.

\subsection{Ordered trees and progress measures.}
\label{sec:pm-and-ut}

An \emph{ordered tree} is a prefix-closed set of sequences of a
linearly ordered set. 
We refer to such sequences as tree \emph{nodes}, we call the elements
of such sequences \emph{branching directions}, and we use the standard
ancestor-descendent terminology for nodes.
For example: node $\seq{}$ is the \emph{root} of a tree; node
$\seq{x}$ is the \emph{child} of the root that is reached from the
root via the branching direction~$x$; node $\seq{x, y}$ is the
\emph{parent} of node $\seq{x, y, z}$; nodes $\seq{}$, $\seq{x}$, and
$\seq{x, y}$ are \emph{ancestors} of node $\seq{x, y, z}$; and nodes
$\seq{x, y}$ and $\seq{x, y, z}$ are \emph{descendants} of nodes
$\seq{}$ and $\seq{x}$. 
Moreover, a node is a \emph{leaf} if it does not have any
descendants.
All nodes in an ordered tree are linearly ordered by the
\emph{lexicographic order} on sequences that is induced by the assumed
linear order on the set of branching directions. 
For example, we have $\seq{x} < \seq{x, y}$, and 
$\seq{x, y, w} < \seq{x, z}$ if $y < z$. 
The \emph{depth} of a node is the number of elements in the eponymous
sequence, the \emph{height} of a tree is the maximum depth of a
node in it, and the \emph{size} of a tree is the number of its
leaves.

A \emph{tree labelling} of a parity game is a mapping $\mu$ from
the vertices in the game graph to leaves in an ordered tree of
height~$d/2$;
for convenience and without loss of generality we assume that every
leaf has depth~$d/2$.   
We write $\seq{m_{d-1}, m_{d-3}, \dots, m_1}$ to denote such a leaf,
and for every priority~$p$, $1 \leq p \leq d$, we define its
$p$-truncation $\seq{m_{d-1}, m_{d-3}, \dots, m_1}|_p$ to be the
sequence $\seq{m_{d-1}, m_{d-3}, \dots, m_p}$ if $p$ is odd, and 
$\seq{m_{d-1}, m_{d-3}, \dots, m_{p+1}}$ if $p$ is even. 
We say that a tree labelling~$\mu$ of the game is a 
\emph{progress measure} if the following \emph{progress condition}
holds for every edge $(v, u)$ in the strategy subgraph of some
positional strategy for Even: 
\begin{compactitem}
\item
  if $p$ is even then $\mu(v)|_{\pi(v, u)} \geq \mu(u)|_{\pi(v, u)}$;
\item
  if $p$ is odd then $\mu(v)|_{\pi(v, u)} > \mu(u)|_{\pi(v, u)}$. 
\end{compactitem}
We recommend inspecting the (brief and elementary) proof
of~\cite[Lemma~2]{JL17}, which establishes that every cycle in the 
strategy subgraph whose all edges satisfy the progress condition is
even. 
It gives a quick insight into the fundamental properties of progress
measures and it shows the easy implication in the following theorem
that establishes progress measures as witnesses of winning strategies
in parity games.  

\begin{theorem}[\cite{EJ91,Jur00}] 
\label{thm:pm-witness-ws}
  Even has a winning strategy from every vertex in a parity game if
  and only if there is a progress measure on the game graph. 
\end{theorem}

\subsection{Finding tree witnesses on universal trees.}
\label{sec:universal-trees}

It is a straightforward but fruitful observation of Jurdzi\'nski and
Lazi\'c~\cite{JL17} that a progress measure on a game graph with~$n$
vertices and at most~$d$ distinct edge priorities is a mapping from
the vertices in the game graph to nodes in an ordered tree  
of height at most~$d/2$ and \emph{with at most~$n$ leaves}
(all subtrees that no vertex is mapped to can be pruned). 
It motivates the fundamental concept that we introduce in this 
section---\emph{universal trees}.

An \emph{$(\ell, h)$-universal (ordered) tree} is an ordered tree, 
such that every ordered tree of height at most~$h$ and
with at most~$\ell$ leaves can be isomorphically embedded into it; 
in such an embedding, the root of the tree must be mapped onto the
root of the universal tree, and the children of every node must be 
mapped---injectively and in an order-preserving way---onto the
children of its image.  

The following proposition follows directly from the above
``straightforward but fruitful'' observation and the definition of a
universal tree.

\begin{proposition}[\cite{JL17}]
\label{prop:pm-into-ut}
  Every progress measure on a graph with~$n$ vertices and with at
  most~$d$ priorities can be seen as a map into an 
  $(n, d/2)$-universal tree. 
\end{proposition}
We offer the following interpretation of the
\emph{small progress measure lifting}~\cite{Jur00} and the
\emph{succinct progress measure lifting}~\cite{JL17} algorithms,
highlighting the central role played in them by the concept of
universal trees that we introduce here.  
Both algorithms perform an iterative search for a progress measure
(a witness for a winning strategy for Even), and their search
spaces are limited by restricting the labels considered in candidate
tree labellings to leaves in specific $(n, d/2)$-universal trees. 
Where the two algorithms differ is the universal trees they use: 
in the former algorithm it is the full $n$-ary tree of height $d/2$, 
and in the latter it is an ordered tree of height $d/2$ in which the
branching directions are bit strings, where for every node, the total
number of bits used in all its branching directions is bounded by
$\ceil{\lg n}$. 

We observe that what is common for both algorithms is that their
correctness relies precisely on the universality property of the
ordered tree from which the candidate labels are taken from.
Indeed, if there is a witness for a winning strategy for Even in the
form of a progress measure
(whose existence is guaranteed by Theorem~\ref{thm:pm-witness-ws}),
then by Proposition~\ref{prop:pm-into-ut}, it suffices to look for one
that uses only leaves of a universal tree as labels.
The original papers~\cite{Jur00,JL17} contain the technical details 
of the iterative scheme that both algorithms use.
Here, we highlight the following two key insights about the design and
analysis of those algorithms:
\begin{compactitem}
\item
  each algorithm computes a sequence of tree labellings that is
  monotonically increasing
  (w.r.t.\ the pointwise lexicographic order),  
  while maintaining the following invariant: 
  the current labelling is pointwise lexicographically smaller than
  or equal to every progress measure;
\item
  the worst-case run-time bound for the algorithm is determined 
  (up to a small polynomial factor) by the size of the universal
  tree used. 
\end{compactitem}

Using our interpretation of the progress measure lifting algorithms
and the terminology of universal trees, the small progress measure
lifting algorithm~\cite{Jur00} is using a tree whose universality is 
straightforward, and whose size---$\Theta(n^{d/2})$---is exponential
in the number of priorities.
On the other hand, the main technical contribution of Jur\-dzi\'nski
and Lazi\'c---that yields their quasi-polynomial succinct progress
measure lifting algorithm~\cite[Theorem~7]{JL17}---can be stated as
the following quasi-polynomial upper bound on the size of universal 
trees. 

\begin{theorem}[{\cite[Lemmas~1 and~6]{JL17}}]
\label{thm:small-universal-tree}
  For all positive integers~$\ell$ and~$h$, there is an  
  $(\ell, h)$-universal tree with at most quasi-polynomial number of 
  leaves. 
  More specifically, the number of leaves is at most 
  $2\ell {\lceil \lg \ell \rceil + h + 1 \choose h}$, which is  
  polynomial in~$\ell$ if $h = O(\log \ell)$; 
  it is $O\left(h \ell^{\lg(h/{\lg \ell})+1.45}\right)$ if
  $h = \omega(\log \ell)$; 
  and---more crude\-ly---it is always $\ell^{\lg h + O(1)}$. 
\end{theorem}
A natural question then arises, whether one can significantly improve
the worst-case upper bounds on the run-time complexity of solving
parity games by designing significantly smaller universal trees.
We give a negative answer to this question in the next section.

\subsection{Smallest universal trees are quasi-poly\-nomial.} 
\label{sec:lower-universal}

The main technical result in this section is a quasi-polynomial lower
bound on the size of universal ordered trees that matches the
upper bound in Theorem~\ref{thm:small-universal-tree} 
up to a small polynomial factor.
It follows that the smallest universal ordered trees have
quasi-polynomial size, and hence the worst-case performance of
progress measure lifting algorithms~\cite{Jur00,JL17} cannot be
improved to sub-quasi-polynomial by designing smaller universal
ordered trees.\footnote{The 
  quasi-polynomial lower bound on the size of universal trees has
  appeared in the technical report~\cite{Fij18}, which is subsumed by
  this paper.}

\begin{theorem}\label{thm:lower_bound_universal_tree}
  For all positive integers~$\ell$ and~$h$, every
  $(\ell, h)$-universal tree has at least
  $\binom{\floor{\lg \ell} + h-1}{h-1}$ leaves, 
  which is at least $\ell^{\lg(h/{\lg \ell}) - 1}$ provided that
  $2h \leq \ell$.
\end{theorem}
This lower bound result shares some similarities with a result of 
Goldberg and Lifschitz~\cite{GL68}, which is for universal trees of a
different kind: 
the height is not bounded and the trees are not ordered.  

\begin{proof}
  First, we give a derivation of the latter bound from the former;
  we show that
  \[
  \binom{\floor{\lg \ell} + h-1}{h-1} = 
  \binom{\floor{\lg \ell} + h-1}{\floor{\lg \ell}} \geq
    \ell^{\lg(h/\lg \ell) - 1}
  \]
  provided that $2h \leq \ell$. 
  We start from the inequality 
  $\left(\frac{k}{\ell}\right)^\ell \le \binom{k}{\ell}$
  applied to the binomial coefficient
  $\binom{\floor{\lg \ell} + h - 1}{\floor{\lg \ell}}$, and take the
  $\lg$ of both sides. 
  This yields
  \begin{multline*}
    \lg \binom{\floor{\lg \ell} + h - 1}{\floor{\lg \ell}} 
    \geq \\
    \geq 
    \floor{\lg \ell} \cdot
    \Big[
      \lg\big(\floor{\lg \ell} + h - 1\big) - \lg \floor{\lg \ell}
    \Big] \\
    \geq 
    (\lg \ell - 1) \cdot
    (\lg h - \lg \lg \ell) \\
    \geq
    \lg \ell \cdot
    \big(\lg (h / \lg \ell) - 1\big), 
  \end{multline*}
  where the second inequality follows since $\ell \geq 2$,
  and the third by the assumption that $2h \leq \ell$.

To prove the first bound, we proceed by induction and show that any
$(\ell, h)$-universal tree has at least $g(\ell, h)$ leaves, where 
\[
g(\ell, h)  =
\sum_{\delta = 1}^{\ell} g(\lfloor \ell / \delta\rfloor, h-1), 
\]
$g(\ell, 1) = \ell$, and $g(1, h) = 1$.

The bounds are clear for $h = 1$ or $\ell = 1$.

Let $T$ be a $(\ell, h)$-universal tree, and
$\delta \in \left\{1, \ldots, \ell\right\}$.   
We claim that the number of nodes at depth $h-1$ of degree greater
than or equal to $\delta$ is at least
$g(\lfloor \ell / \delta \rfloor, h-1)$. 

Let $T_\delta$ be the subtree of $T$ obtained by removing all leaves
and all nodes at depth $h-1$ of degree less than $\delta$: the leaves
of the tree $T_\delta$ have depth exactly $h-1$. 

We argue that $T_\delta$ is
$(\lfloor \ell / \delta \rfloor, h-1)$-universal. 
Indeed, let $t$ be a tree with $\lfloor \ell / \delta \rfloor$ leaves
all at depth $h-1$. 
To each leaf of $t$ we append $\delta$ children, yielding the tree
$t_+$ which has $\lfloor \ell / \delta \rfloor \cdot \delta \le \ell$ 
leaves all at depth~$h$.
Since $T$ is $(\ell, h)$-universal, the tree $t_+$ embeds into $T$.
Observe that the embedding induces an embedding of $t$ into
$T_\delta$, since the leaves of $t$ have degree $\delta$ in $t_+$,
hence are also in $T_\delta$. 

Let $\ell_\delta$ be the number of nodes at depth $h-1$ with degree
exactly $\delta$. 
So far we proved that the number of nodes at depth $h-1$ 
of degree greater than or equal to $\delta$ is at least
$g(\lfloor \ell / \delta \rfloor, h-1)$, 
so 
\[
\sum_{i=\delta}^{\ell} \ell_i \geq g(\lfloor \ell / \delta \rfloor, h-1).
\]
Thus the number of leaves of $T$ is 
\[
  \sum_{i=1}^\ell \ell_i \cdot i =
  \sum_{\delta=1}^{\ell} \sum_{i=\delta}^{\ell} \ell_i \geq
  \sum_{\delta=1}^{\ell} g(\lfloor \ell / \delta \rfloor, h-1) =
  g(\ell, h).  
\]

It remains to prove that
\begin{equation}
\label{equation:glh-binom}
  g(\ell, h) \ge \binom{\floor{\lg \ell} + h - 1}{\floor{\lg \ell}}.
\end{equation}
Define $G(p, h) = g(2^p, h)$ for $p \ge 0$ and $h \ge 1$.
Then we have
\[
  \begin{array}{lll}
    G(p, h) & \ge & \sum_{k=0}^p G(p-k, h-1), \\
    G(p, 1) & \ge & 1, \\
    G(0, h) & = & 1.
  \end{array}
\]
To obtain a lower bound on $G$ we define $\overline{G}$ by
\[
  \begin{array}{lll}
    \overline{G}(p, h) & = & \overline{G}(p, h-1) + \overline{G}(p-1, h), \\
    \overline{G}(p, 1) & = & 1, \\
    \overline{G}(0, h) & = & 1,
  \end{array}
\]
so that $G(p, h) \ge \overline{G}(p, h)$.
We verify by induction that $\overline{G}(p,h) = \binom{p+h-1}{p}$,
which follows from Pascal's identity  
\[
\binom{p+h-1}{p} = \binom{p+h-2}{p} + \binom{p+h-2}{p-1}.
\]
This implies that $G(p, h) \ge \binom{p+h-1}{p}$, from
which~(\ref{equation:glh-binom}) follows. 
\end{proof}

\section{The separation approach}
\label{sec:prelim}

\subsection{Languages of play encodings.} 
\label{sec:langs}

The outcome of the two players interacting in a parity game by making
moves is an infinite path in the game graph.  
We encode such infinite paths as infinite words over the alphabet
$\Sigma_d = \set{1, 2, \dots, d}$ in a natural way:
each move---in which an edge~$e$ is followed---is encoded by the
letter $\pi(e)$, i.e., the priority of edge~$e$. 

We write $\limsupeven_d$ for the set of infinite words 
in which the biggest number that
occurs infinitely many times is even, and we write $\limsupodd_d$
for the set of infinite words in which that number is odd. 
Observe that sets $\limsupeven_d$ and $\limsupodd_d$ form a
partition of the set $(\Sigma_d)^\omega$ of all infinite words over the
alphabet~$\Sigma_d$. 
As intended, an infinite play in a parity game graph 
(of arbitrary size) 
with edge priorities not exceeding~$d$ is winning for Even if and only
if the infinite-word encoding of the play is in~$\limsupeven_d$. 

Recall that a (parity game) graph is called even if every cycle in it
is even 
(i.e., the highest priority that occurs on the cycle is even).  
For all positive integers~$n$ and~$d$, we define the language 
$\evencycles_{n, d} \subseteq (\Sigma_d)^\omega$ to consist of
infinite words that encode an infinite path in an even graph with at
most~$n$ vertices and~$d$ priorities.  
The languages $\evencycles_{n, d}$ can be thought of as finitary 
under-approximations of the language $\limsupeven_d$ because  
\[
  \evencycles_{1, d} \subsetneq \evencycles_{2, d} \subsetneq \cdots
  \subsetneq \limsupeven_d. 
\]
Languages $\oddcycles_{n, d}$---that can be thought of as finitary 
under-approximations of the language $\limsupodd_d$---are defined in
an analogous way.

\subsection{Safety automata and games. } 
\label{sec:sep-approach}

The fundamental and simple model that the statement of our main
technical result formally refers to is a (non-deterministic)
\emph{safety automaton}.   
Superficially, it closely resembles the classic model of 
\emph{finite automata}: each safety automaton has a finite set of
\emph{states}, a designated \emph{initial state}, and a 
\emph{transition relation}. 
(Without loss of generality, we assume that the transition relation is
\emph{total}, i.e., for every state~$s$ and letter~$a$, there is at
least one state~$s'$, such that the triple $(s, a, s')$ is in the
transition relation.)
The differences between our model of safety automata and the classic
model of finite automata with designated \emph{accepting states} are
as follows: 
\begin{compactitem}
\item
  safety automata are meant to accept or reject 
  \emph{infinite words}, not finite ones; 
\item
  a safety automaton does not have a designated set of accepting
  states; 
  instead it has a designated set of \emph{rejecting states};
\item
  a safety automaton \emph{accepts} an infinite word if there is an
  infinite run of the automaton on the word in which no rejecting
  state ever occurs;
  otherwise it \emph{rejects} the infinite word. 
\end{compactitem}
We say that a safety automaton is \emph{deterministic} if the
transition relation is a function: 
for every state~$s$ and letter~$a$, there is a unique state~$s'$, such
that the triple $(s, a, s')$ is a transition. 

Finally, we define the elementary concept of \emph{safety games},
which are played by two players on finite directed graphs in a similar
way to parity games, but the goals of the players are simpler than in
parity games: 
the \emph{safety player} wins if a designated set of 
\emph{unsafe vertices} is never visited, and otherwise the opponent
(sometimes called the \emph{reachability player}) wins.

\subsection{Safety separating automata.}
\label{sec:simple-separator}

For positive integers $n$ and~$d$, we say that an automaton~$\A$ is an
\emph{$(n, d)$-separator} if it is a separator of $\evencycles_{n, d}$
and $\oddcycles_{n, d}$; 
and we say that it is a \emph{strong $(n, d)$-separator} if it is a
separator of $\evencycles_{n, d}$ and $\limsupodd_d$. 
Note that if an automaton is a strong $(n, d)$-separator then it is
also an $(n, d)$-separator, but not every $(n, d)$-separator is
strong. 

To illustrate the concept of a (strong) $(n, d)$-separator, 
we present a simple ``multi-counter'' strong $(n, d)$-separator that
is implicit in the work of Bernet et al.~\cite{BJW02}.  
We define the automaton~$\C_{n, d}$ that, for every odd priority~$p$,  
$1 \leq p \leq d-1$, keeps a counter $c_p$ that stores the number of
occurrences of priority~$p$ since the last occurrence of a priority
larger than~$p$ (even or odd). 
It is a safety automaton: it rejects a word immediately once the
integer stored in any of the $d/2$ counters exceeds~$n$. 

In fact, instead of ``counting up'' (from~$0$ to~$n$), we prefer to
``count down'' (from~$n$ to~$0$), which is equivalent, but it aligns
better with the definition of progress measures. 
More formally, we define the deterministic safety 
automaton~$\C_{n, d}$ in the following way:
\begin{compactitem}
\item
  the set of states of~$\C_{n, d}$ is the set of $d/2$-sequences 
  $\seq{c_{d-1}, c_{d-3}, \dots, c_1}$, such that $c_p$ is an integer
  such that $0 \leq c_p \leq n$ for every odd $p$, 
  $1 \leq p \leq d-1$;  
  and it also contains an additional \emph{rejecting}
  state~$\mathtt{reject}$; 
\item
  the initial state is $\seq{n, n, \dots, n}$; 
\item
  the transition function~$\delta$ is defined so that
  $\delta\big(\mathtt{reject}, p\big) = \mathtt{reject}$ for
  all $p \in \Sigma_d$, and
  $\delta\big(\seq{c_{d-1}, c_{d-3}, \dots, c_1}, p\big)$ is equal to:
  \begin{compactitem}
  \item
    $\seq{c_{d-1}, c_{d-3}, \dots, c_{p+1}, n, \dots, n}$ if $p$ is
    even, 
  \item
    $\seq{c_{d-1}, c_{d-3}, \dots, c_p-1, n, \dots, n}$
    if $p$ is odd and $c_p > 0$,
  \item
    $\mathtt{reject}$ 
    if $p$ is odd and $c_p = 0$. 
  \end{compactitem}
\end{compactitem}
Note that the size of automaton~$\C_{n, d}$ is $\Theta(n^{d/2})$.

\begin{proposition}[\cite{BJW02,BC18}] 
\label{prop:simple-separator}
  The safety automaton $\C_{n, d}$ is a strong $(n, d)$-separator. 
\end{proposition}

\begin{proof}
  Firstly, we argue that if the unique run of~$\C_{n, d}$ on an infinite
  word contains an occurrence of the rejecting state then the word is
  not in $\evencycles_{n, d}$.
  Indeed, the only reason for the unique run of~$\C_{n, d}$ to reach the
  rejecting state is that the state reached after reading some prefix of 
  the word is 
  $\seq{c_{d-1}, c_{d-3}, \dots, c_1}$ with $c_p = 0$ for an odd~$p$,
  and~$p$ is subsequently read. 
  For this to happen, there must be a suffix of
  the prefix in which there are~$n$ occurrences of priority~$p$ and no
  priority higher than~$p$ occurs, and the currently read letter is the
  $(n+1)$-st occurrence of priority~$p$ in the prefix.
  We argue that if the input word is an encoding of an infinite play
  in an even graph with at most~$n$ vertices then---by the pigeonhole
  principle---there is a cycle in the graph in which the highest
  priority is~$p$, which contradicts the assumption that the graph was
  even. 
  It follows that the infinite word is not in~$\evencycles_{n, d}$. 

  Secondly, we argue that if a word is in $\limsupodd_d$ then the
  unique run of~$\C_{n, d}$ on the word contains an occurrence of the
  rejecting state. 
  Consider an infinite suffix of the word in which all priorities occur
  infinitely many times. 
  Unless the unique run reached the rejecting state on the corresponding
  prefix already, let $\seq{c_{d-1}, c_{d-3}, \dots, c_1}$ be the state
  reached in the unique run at the end of the prefix.  
  By the assumption that the word is in $\limsupodd_d$, the highest
  priority that occurs in the suffix is odd and it occurs infinitely
  many times.   
  Take the shortest prefix of the suffix in which the highest priority
  $p$ occurs $n - c_p + 1$ times.
  The unique run of~$\C_{n, d}$ on the original infinite word reaches
  the rejecting state upon reading that prefix. 
\end{proof}

\subsection{The separation approach.}
We now explain how safety separating automata 
allow to reduce the complex task of 
solving a parity game to the (conceptually and algorithmically)
straightforward task of solving a safety game, by exploiting
positional determinacy of parity games.   
This is the essence of the \emph{separation approach} that implicitly
underpins the algorithms of Bernet, Janin, and
Walukiewicz~\cite{BJW02} and of Calude et al.~\cite{CJKLS17}, as
formalized by Boja\'nczyk and 
Czerwi\'nski~\cite[Chapter 3]{BC18}.
Here, we only consider the simple case of \emph{deterministic}
automata. 
We postpone the discussion of using non-deterministic automata in the
separation approach to
Sections~\ref{section:non-deterministic-separators}
and~\ref{section:separator-from-register-games}, which is the  
only place where non-determinism seems to be needed. 

Given a parity game~$G$ with at most $n$ vertices and priorities up
to~$d$, and a deterministic safety automaton~$\A$ 
with input alphabet $\Sigma_d$,
we define a safety game as the \emph{chained product} $G \rhd \A$, in
which 
\begin{compactitem}
\item
  the dynamics of play and ownership of vertices is inherited from
  the parity game~$G$; 
\item
  the automaton~$\A$ is fed the priorities of the edges corresponding
  to the moves made by the players;
\item
  the safety winning condition is the safety acceptance condition of
  the automaton~$\A$. 
\end{compactitem}

\begin{proposition}
\label{prop:safety-game-equiv}
  If $G$ is a parity game with~$n$ vertices and priorities up to~$d$,
  and if $\A$ is a deterministic safety $(n, d)$-separator,  
  then Even has a winning strategy in~$G$ if and only if she has a
  winning strategy in the chained-product safety 
  game~$G \rhd \A$.      
\end{proposition}

\begin{proof}
  If Even has a winning strategy in the parity game~$G$
  then---by positional determinacy~\cite{EJ91,Mos91}---she also has one
  that is positional. 
  We argue that if Even uses such a positional winning strategy in~$G$
  when playing the chained-product game~$G \rhd \A$, then she
  also wins the latter. 
  Indeed, the strategy subgraph of the positional winning strategy for
  Even in~$G$ is an even graph with at most~$n$ vertices, and hence
  all words that are fed to the automaton~$\A$ are in 
  $\evencycles_{n, d}$, and hence they are accepted by~$\A$.  
  
  Otherwise, Odd has a positional winning strategy in~$G$, and it can
  be transferred to a winning strategy for him in~$G \rhd \A$ in the
  same way as we argued for Even above. 
\end{proof}

In the rest of the paper, we focus on \emph{strong} 
$(n, d)$-separators for two reasons.  
Firstly---as described in Section~\ref{sec:seps-everywhere}---all
the known quasi-polynomial algorithms for parity games are
underpinned by strong $(n, d)$-separators. 
Secondly, our proof of the quasi-polynomial lower bound in
Section~\ref{sec:proof} only applies to strong $(n, d)$-separators.

\section{Separating automata everywhere}
\label{sec:seps-everywhere}

In this section we argue that the three distinct techniques that have
been developed so far for designing quasi-polynomial algorithms for
solving parity games can be unified as instances of the separation
approach introduced in the previous section. 
The main unifying aspect that we highlight in this section is that all
the three approaches yield constructions of separating automata
of quasi-polynomial size, which provides evidence of significance of
our main technical result:
the quasi-polynomial lower bound on the size of separators
(Theorem~\ref{thm:lower-bound}) forms a barrier that all of those
approaches need to overcome in order to further significantly improve
the known complexity of solving parity games. 
We note that, in contrast to the results of Calude et
al.~\cite{CJKLS17} and Lehtinen~\cite{Leh18}, not all of the proposed
quasi-polynomial algorithms explicitly construct separating automata
or other objects of (at least) quasi-polynomial
size~\cite{JL17,FJSSW17}, but in the worst case, they too enumerate
structures that form the states of the related separating automata
(leaves in a universal tree and play summaries, respectively). 

In Section~\ref{section:separator-from-universal-trees} we generalize
the simple separator described in Section~\ref{sec:simple-separator}
to a construction that creates strong separators from arbitrary
universal trees.
This result implies that the main combinatorial structure underlying 
all the progress measure lifting
algorithms~\cite{Jur00,BJW02,Sch17,CHL17}, including the
quasi-polynomial succinct progress measure lifting of Jurdzi\'nski and
Lazi\'c~\cite{JL17,DJL18}, is intimately linked to the separation
approach. 
In Section~\ref{section:separator-from-play-summaries}, we briefly
discuss the observation of Boja\'nczyk and Czerwi\'nski~\cite{BC18}
that Calude et al.'s~\cite{CJKLS17} play summaries construction can be
straightforwardly interpreted as defining a separating automaton; 
we refer the reader to their very readable technical exposition.  
Finally, in Sections~\ref{section:non-deterministic-separators}
and~\ref{section:separator-from-register-games}, we 
discuss how to adapt the separation approach to also encapsulate the
most recent quasi-polynomial algorithm for solving parity games by
Lehtinen~\cite{Leh18}, based on register games.
This requires care because, unlike the constructions based on play
summaries and universal trees, separating automata that underpin
Lehtinen's reduction from parity games to register games seem to
require non-determinism, and in general,
Proposition~\ref{prop:safety-game-equiv} does not hold for
non-deterministic automata.

\subsection{Separating automata from universal trees.}
\label{section:separator-from-universal-trees}

For positive integers~$n$ and~$d$, such that $d$ is even, let 
$L_{n, d/2}$ be the set of leaves in an $(n, d/2)$-universal tree. 
The definition of the deterministic safety automaton~$\U_{n, d}$ bears
similarity to the definition of the simple ``multi-counter''
separator~$\S_{n, d}$ from Section~\ref{sec:simple-separator}.  
Again, the states are $d/2$-sequences of ``counters'', but the
``counting down'' is done in a more abstract way than in~$\S_{n, d}$,
using instead the natural lexicographic order on the nodes of the
universal tree.   

More formally, we define a deterministic safety automaton~$\U_{n, d}$
in the following way: 
\begin{compactitem}
\item
  the set of states of~$\U_{n, d}$ is the set $L_{n, d/2}$ of leaves
  in the $(n, d/2)$-universal tree together with an additional 
  $\mathtt{reject}$ state;
\item
  the initial state is the largest leaf 
  (in the lexicographic tree order)
  and the only rejecting state is $\mathtt{reject}$;  
\item
  the transition function~$\delta$ is defined so that
  $\delta(\mathtt{reject}, p) = \mathtt{reject}$ for all
  $p \in \Sigma_d$, and if $s \not= \mathtt{reject}$ then
  $\delta(s, p)$ is equal to:
  \begin{compactitem}
  \item
    the largest leaf $s'$ such that $s|_p = s'|_p$ if $p$ is even,
  \item
    the largest leaf $s'$ such that $s|_p > s'|_p$ if $p$ is odd,
  \item
    $\mathtt{reject}$ if $p$ is odd and no such leaf exists. 
  \end{compactitem}
\end{compactitem}

\begin{remark}
  Let $\T_{n, d}$ be an ordered tree in which all the leaves have
  depth~$d/2$ and every non-leaf node has exactly $n$ children;
  note that $\T_{n, d}$ is trivially an $(n, d/2)$-universal tree. 
  An instructive exercise that we recommend is to compare the structures
  of the automaton~$\U_{n, d}$ based on the $(n, d/2)$-universal
  tree~$\T_{n, d}$ and of the simple separating automaton~$\S_{n, d}$
  from Section~\ref{sec:simple-separator}.
\end{remark}

\begin{theorem}
\label{thm:sep-from-ut}
  For every $(n, d/2)$-universal tree, 
  the safety automaton~$\U_{n, d}$ is a strong $(n, d)$-separator.  
\end{theorem}

\begin{proof}
  First we prove that if an infinite word $w \in \Sigma_d$ is
  in~$\evencycles_{n, d}$ then it is accepted by the safety
  automaton~$\U_{n, d}$.  
  Let $G$ be an even graph with at most~$n$ vertices and priorities up
  to~$d$, in which~$w$ occurs as an encoding of an infinite path;
  let $\seq{v_1, v_2, v_3, \dots}$ be the sequence of vertices in the
  infinite path in~$G$. 
  By Theorem~\ref{thm:pm-witness-ws} and
  Proposition~\ref{prop:pm-into-ut}, there is a progress measure~$\mu$
  that maps vertices in~$G$ into the set of leaves~$L_{n, d/2}$ in the
  $(n, d/2)$-universal tree. 
  Consider the unique run $\seq{q_1, q_2, q_3, \dots}$ of~$\U_{n, d}$
  on the word~$w$, where $q_1$ is the initial state of~$\U_{n, d}$ and
  we have $\delta(q_i, \pi(v_i, v_{i+1})) = q_{i+1}$ for all positive
  integers~$i$. 
  By the definition of the initial state in~$\U_{n, d}$, we have that
  $q_1 \geq \mu(v_1)$, and the definition of the transition function
  of~$\U_{n, d}$ allows to establish---by a straightforward induction 
  on~$i$---that $q_i \geq \mu(v_i)$ for all positive integers~$i$.  
  It follows that the run never reaches the $\mathtt{reject}$ state,
  and hence the run is accepting. 

  Now we argue that for every infinite word~$w$ in~$\limsupodd_d$, the 
  unique run of~$\U_{n, d}$ on it eventually reaches the
  $\mathtt{reject}$ state. 
  Our argument is a generalization of the corresponding one in the
  proof of Proposition~\ref{prop:simple-separator}.  
  Consider an infinite suffix~$x = \seq{x_1, x_2, x_3, \dots}$ of the
  word~$w$ in which all priorities occur infinitely many times. 
  Unless the unique run reached the rejecting state on the
  corresponding prefix already, let~$s_1$ be the state reached in the 
  unique run at the end of the prefix. 
  Let $p$ be the highest priority that occurs in~$x$;
  by the assumption that~$w$ is in~$\limsupodd_d$, the priority~$p$ is
  odd. 
  Let $\seq{s_1, s_2, s_3, \dots}$ be the unique run of~$\U_{n, d}$
  on~$x$: for every positive integer~$i$, we have 
  $\delta(s_i, x_i) = s_{i+1}$. 
  Note that $\seq{s_1, s_2, s_3, \dots}$ is a suffix of the unique run  
  of~$\U_{n, d}$ on~$w$. 
  By the definition of the transition function of~$\U_{n, d}$, if
  $s_i$ and $s_{i+1}$ are not equal to the $\mathtt{reject}$ state
  then $s_i|_p \geq s_{i+1}|_p$, and if $x_i = p$ then 
  $s_i|_p > s_{i+1}|_p$. 
  Therefore, unless $s_i = \mathtt{reject}$ for some positive
  integer~$i$, there is an infinite subsequence of
  $\seq{s_1, s_2, s_3, \dots}$ whose $p$-truncations form a strictly
  decreasing sequence, which contradicts finiteness of the 
  set~$L_{n, d/2}$ of states of~$\U_{n, d}$. 
\end{proof}

The following can be obtained by using the nearly optimal
quasi-polynomial $(n, d/2)$-universal trees from 
Theorem~\ref{thm:small-universal-tree} in the construction of the 
automaton $\U_{n, d}$.

\begin{corollary}
  There are deterministic safety auto\-ma\-ta that are strong
  $(n, d)$-separators of size~$n^{\lg d + O(1)}$.  
\end{corollary}

\subsection{Separating automata from play summaries.}
\label{section:separator-from-play-summaries}

Boja\'nczyk and Czerwi\'nski~\cite[Chapter 3]{BC18} give an accessible
exposition of how the breakthrough result of Calude et
al.~\cite{CJKLS17} can be viewed as a construction of a deterministic 
automaton of quasi-polynomial size that separates 
$\evenloops_{n, d}$ from $\limsupodd_d$, which implies separation of
$\evencycles_{n, d}$ from~$\limsupodd_d$.  

One superficial difference between our exposition of separators and
theirs is that we use the model of safety automata, while they
consider the dual model of \emph{reachability automata} instead.
(In reachability automata, an infinite word is accepted if and only if
one of the designated \emph{accepting states} is reached; otherwise it
is rejected.)
If, in Boja\'nczyk and Czerwi\'nski's construction, we swap the roles
of players Even and Odd, and we make the accepting states rejecting,
we get a safety automaton that separates $\evencycles_{n, d}$  
from $\limsupodd_d$. 

\begin{theorem}[\cite{CJKLS17,BC18}] 
  The play summaries data structure of Calude et al.\ yields
  deterministic safety automata that are strong $(n, d)$-separators 
  of size $n^{\lg d + O(1)}$. 
\end{theorem}

\subsection{Non-deterministic automata and the separation approach.}
\label{section:non-deterministic-separators}

The possible usage of non-deterministic automata in the separation
approach to solving parity games is less straightforward.
First of all, the game dynamics needs to be modified to explicitly
include the choices that resolve non-determinism in every step. 
We give the power to make those choices to Even, but this extra power
does not suffice to make her win the chained-product game
whenever she has a winning strategy in the original parity game. 
The reason for this failure in transferring winning strategies from
the parity game to the chained-product safety game is that in
arbitrary non-deterministic automata it may be impossible to
successfully resolve non-deterministic choices at a position in the
input word without knowing the letters at the later positions.  
In the game, however, the play is a result of the future choices 
of both the player and her opponent, and the former cannot predict the
latter.

A well-known example of non-deterministic automata for which the 
chained-product game is equivalent to the original game is the
class of \emph{good-for-games} automata~\cite{HP06}. 
They have the desired property that the non-deterministic
choices of the automaton can always be resolved based only on the
letters in the word at the positions up to the current one, thus
making it possible to continue constructing an accepting run for all
words accepted by the non-deterministic automaton that have the word
read so far as a prefix.

We consider a weaker property than the one mentioned above, one that 
is sufficient for strategy transfer in the context of the separation
approach.  
We say that a non-deterministic automaton~$\A$ is a
\emph{good-for-separation} strong $(n, d)$-separator if for every
parity game~$G$ with at most~$n$ vertices and~$d$ priorities, in the 
chained-product game $G \rhd \A$, there is a way to resolve the
non-deterministic choices of~$\A$ based only on the prefix of the play
so far, in such a way that a construction of an accepting run can be
continued for all words in $\evencycles_{n, d}$ 
(but not necessarily for all the words accepted by~$\A$). 

More formally, a non-deterministic automaton~$\A$ is a
\emph{good-for-separation} strong $(n, d)$-separator if it rejects
every word in~$\limsupodd_d$, and for every even game~$G$ with~$n$
vertices and priorities up to~$d$, there is a winning strategy for
Even in the chained-product game~$G \rhd \A$. 
We note that the automata that can be derived using Lehtinen's
techniques~\cite{Leh18} are good-for-separation and hence they are
suitable for being used in the separation approach.  

Note that our quasi-polynomial lower bound in Section~\ref{sec:proof}
holds for all non-deterministic strong $(n, d)$-separators, regardless
of their suitability for strategy transfer and for solving parity
games using the separation approach.

\subsection{Separating automata from register games.}
\label{section:separator-from-register-games}

First, we recall the definition of a (non-deterministic) 
\emph{parity automaton}. 
Like safety automata, parity automata process infinite words, but
instead of having designated rejecting states, every transition has a
\emph{priority}, which is a positive integer.  
The set of transitions consists of tuples of the form $(s, a, p, s')$:
such a transition has priority~$p$, and when reading a letter~$a$ in
state~$s$, the automaton moves to state~$s'$. 
We say that a run of a parity automaton on an infinite word is
\emph{accepting} if the highest transition priority that occurs
infinitely often is even; otherwise it is \emph{rejecting}.  
The automaton \emph{recognizes} the language of all words on which it
has an accepting run. 
The following observation is straightforward.
\begin{proposition}
  For every positive integer~$d$, there is a deterministic parity
  automaton $\Pp_d$, with $1$ state and~$d$ priorities, that
  \emph{recognizes} the language~$\limsupeven_d$. 
\end{proposition}
Indeed, it suffices to equip the transition relation~$\Pp_d$ with 
transitions $(s, p, p, s)$, where $s$ is the unique state of~$\Pp_d$,
for all $p \in \Sigma_d$. 
Observe that \emph{recognizing} the language~$\limsupeven_d$ implies
being a strong $(n, d)$-separator for \emph{all} positive
integers~$n$:  
a much stronger property than being a strong $(n, d)$-separator for
\emph{some} positive integer~$n$. 
Since the automaton~$\Pp_d$ is deterministic, it is in principle
suitable for the separation approach applied to games with~$d$
distinct priorities. 
Note, however, that using it brings no tangible benefit for solving
parity games:   
the chained-product game $G \rhd \Pp_d$ has as few vertices as
the original game~$G$, but its number of distinct priorities is no
smaller.  

We argue that the key technical result in the recent work of
Lehtinen~\cite{Leh18} can be interpreted as proving the following
theorem. 
\begin{theorem}[\cite{Leh18}]
\label{theorem:lehtinen}
  For all positive integers $n$ and~$d$, there is 
  a good-for-separation parity automaton $\Rr_{n, d}$, with at most 
  ${d+\floor{\lg n}+1 \choose d} = n^{\lg d + O(1)}$ states and  
  $2\floor{\lg n}+3$ priorities, that 
  is a strong $(n, d)$-separator. 
\end{theorem}
Note that this theorem does indeed bring tangible benefits for solving
parity games (via the separation approach) because if a parity
game~$G$ has $n$ vertices and~$d$ priorities, then the chained-product
game $G \rhd \Rr_{n, d}$ has only a quasipolynomial number of
vertices and a logarithmic number of priorities. 
Even if an unsophisticated algorithm---with run-time that is
exponential in the number of priorities---is used to solve the
chained-product game $G \rhd \Rr_{n, d}$, the overall run-time
is quasi-polynomial.   

For every parity game~$G$, Lehtinen defines the corresponding
\emph{register game}~$R_G$, whose vertices consist of vertices of
game~$G$ together with $\floor{1+\lg n}$-sequences 
$\seq{r_{\floor{1+\lg n}}, \dots, r_2, r_1}$ of the so-called
\emph{registers} that hold priorities, i.e., numbers from the  
set~$\set{1, 2, \dots, d}$. 
The game is played on a copy of~$G$ in the usual way, additionally at
her every move player Even is given a chance---but not an
obligation---to ``reset'' one of the registers, and each register
always holds the biggest priority that has occurred since it was last
reset.  

What needs explaining is what ``resetting a register'' entails.  
When the register at position~$k$
is reset, then the next register sequence is
$\seq{r_{\floor{1+\lg n}}, \dots, r_{k+1}, r_{k-1}, \dots, r_1, 1}$, that
is registers at positions~$1$ to~$k-1$ are promoted to positions~$2$
to~$k$, respectively, and the just-reset register is now at
position~$1$ and it has value~$1$. 
Moreover---and very importantly for Lehtinen's
construction---resetting the register at position~$k$ causes the even 
priority~$2k$ to occur in game~$R_G$ if the value in the register was
even, and the odd priority~$2k+1$ otherwise. 
If, instead, Even decides not to reset any register then the odd
priority~$1$ occurs. 

Lehtinen's main technical result is that the original parity game~$G$ and
the register game~$R_G$ have the same winners. 
She proves it by arguing that if Even has a (positional) winning
strategy in~$G$ then she has a strategy of resetting registers
in~$R_G$ so that she again wins the parity condition 
(albeit with the number of priorities reduced from an arbitrarily
large~$d$ in~$G$ to only~$\floor{1+\lg n}$ in~$R_G$). 
Our approach is to separate the graph structure from Lehtinen's
mechanism to capture the original parity winning condition using
registers.  
We now define an automata-theoretic analogue of her construction in
which we use non-determinism to model the ability to pick various
resetting strategies.

For all positive numbers~$n$ and~$d$, such that~$d$ is even, we define
the \emph{non-deterministic parity automaton}~$\Rr_{n, d}$ in the
following way. 
\begin{compactitem}
\item
  The set of states of~$\Rr_{n, d}$ is the set of 
  non-increasing $\floor{1+\lg n}$-sequences 
  $\seq{r_{\floor{1+\lg n}}, \dots, r_2, r_1}$ of ``registers'' that hold
  numbers in $\set{1, 2, \dots, d}$.  
  The initial state is $\seq{1, 1, \dots, 1}$. 

\item
  For every state 
   $s = \seq{r_{\floor{1+\lg n}}, \dots, r_2, r_1}$ 
  and
  letter $p \in \Sigma_d$, we define the 
  \emph{update of~$s$ by~$p$} to be the state
  $\seq{r_{\floor{1+\lg n}}, \dots, r_k, p, \dots, p}$, where $k$ is the
  smallest such that $r_k > p$.   

\item
  For every state 
   $s = \seq{r_{\floor{1+\lg n}}, \dots, r_2, r_1}$ 
  and for every~$k$, $1 \leq k \leq \floor{1+\lg n}$, we define the  
  \emph{$k$-reset of~$s$} to be the state 
  $\seq{r_{\floor{1+\lg n}}, \dots, r_{k+1}, r_{k-1}, \dots, r_2, 1}$. 

\item
  For every state~$s$
  and letter $p \in \Sigma_d$, there is a transition  
  $(s, p, 1, s')$ of priority~$1$ in the transition relation, 
  and where $s'$ is the update of~$s$ by~$p$;
  we call this a \emph{non-reset} transition. 

\item
  For every state~$s$, 
  letter $p \in \Sigma_d$, 
  and for every~$k$, $1 \leq k \leq \floor{1+\lg n}$, 
  if $s'' = \seq{r_{\floor{1+\lg n}}, \dots, r_2, r_1}$ is the update
  of~$s$ by~$p$ and $r_k$ is even, 
  then there is a transition $(s, p, 2k, s')$ of priority~$2k$ in the 
  transition relation, where $s'$ is the $k$-reset of~$s''$. 
  We say that this transition is an \emph{even reset of register~$k$}. 
\item
  For every state~$s$,
  letter $p \in \Sigma_d$, 
  and for every~$k$, $1 \leq k \leq \floor{1+\lg n}$, 
  if $s'' = \seq{r_{\floor{1+\lg n}}, \dots, r_2, r_1}$ is the update
  of~$s$ by~$p$ and $r_k$ is odd, 
  then there is a transition $(s, p, 2k+1, s')$ of priority $2k+1$ in
  the transition relation, where $s'$ is the $k$-reset of~$s''$. 
  We say that this transition is an \emph{odd reset of register~$k$}. 
\end{compactitem}
The states of~$\Rr_{n, d}$ are monotone sequences of length $1+\lg n$,
consisting of positive integers no larger than~$d$, and hence their
number is bounded by 
${d+\floor{\lg n}+1 \choose d} \leq d^{\floor{1+\lg n}}$, which is 
$n^{\lg d + O(1)}$.  
We note that the register game $R_G$ that Lehtinen constructs from a
parity game~$G$ is essentially the same as the chained-product
game $G \rhd \Rr_{n, d}$.  

The proof that the parity automaton~$\Rr_{n, d}$ is a strong 
$(n, d)$-separator can be obtained by adapting Lehtinen's proofs
showing that the \emph{register index} of a parity game with~$n$
vertices is at most~$1+\lg n$~\cite[Theorem~4.7]{Leh18}
and that a winning strategy for Odd in a parity game~$G$ yields a
winning strategy for him in the corresponding register  
game~$R_G$~\cite[Lemma~3.3]{Leh18}.    
We refer an interested reader to Lehtinen's paper for the details.
We note, however, that the strategy for resetting registers in the
register game~$R_G$ (and hence also in the chained-product
game~$G \rhd \Rr_{n, d}$)
that Lehtinen constructs in the proof of~\cite[Theorem~4.7]{Leh18}
depends only on the current content of the registers and the current
vertex in the game. 
This property implies that the automata~$\Rr_{n, d}$ are
good-for-separation $(n, d)$-separators in the sense discussed in 
Section~\ref{section:non-deterministic-separators}.

We stress that the quasi-polynomial lower bound for the size of strong 
separators that we establish in Section~\ref{sec:proof} applies to
\emph{safety} automata, but not to \emph{parity} automata such  
as~$\Rr_{n, d}$. 

On the other hand, we argue that the parity automata~$\Rr_{n, d}$ can
be turned into good-for-separation \emph{safety} $(n, d)$-separators
by taking a chained product with other deterministic safety
separators.  
Let~$\Rr$ be a parity automaton over the alphabet~$\Sigma_d$ and
with~$d'$ transition priorities, and let~$\S$ be a safety automaton
over the alphabet~$\Sigma_{d'}$.  
Consider the following chained-product automaton $\Rr \rhd \S$.  
\begin{compactitem}
\item
  The set of states is the set of pairs $(r, s)$, where~$r$ is a state
  of~$\Rr$ and $s$ is a state of~$\S$. 
  The initial state is the pair of the initial states of the two
  automata. 
  A state $(r, s)$ is rejecting if~$s$ is rejecting in~$\S$. 

\item
  If $(r, a, p, r')$ is a transition in~$\Rr$, where 
  $a \in \Sigma_{d}$ is the letter read by the transition and 
  $p \in \Sigma_{d'}$ is the priority of the transition, and if 
  $(s, p, s')$ is a transition in~$\S$  
  (for deterministic automata, we often write 
  $\delta(s, p) = s'$ instead), 
  then $\big((r, s), a, (r', s')\big)$ is a transition 
  in~$\Rr \rhd \S$. 
\end{compactitem}
Let $n' = {d+\floor{\lg n}+1 \choose d} = n^{\lg d + O(1)}$ be the
number of states in the parity register automaton~$\Rr_{n, d}$ and
let $d' = 2\floor{\lg n}+4$ be an upper bound on its priorities.  

\begin{proposition}
\label{proposition:parity-to-safety-separators}
  If $\S$ is a deterministic safety strong $(n', d')$-separator then
  the chained-product automaton $\Rr_{n, d} \rhd \S$ is a safety 
  good-for-separation strong $(n, d)$-separator.  
\end{proposition}

\begin{proof}
  Let $G$ be an even graph with~$n$ vertices and priorities up
  to~$d$. 
  Careful inspection of Lehtinen's proof of~\cite[Theorem~4.7]{Leh18}
  reveals that Even has a \emph{positional} winning strategy~$\sigma$
  in the chained-product game~$G \rhd \Rr_{n, d}$. 
  Note that the strategy subgraph~$G'$ of that strategy is even, and
  that it has at most~$n'$ vertices and priorities up to~$d'$. 
  Since~$\S$ is a deterministic $(n', d')$-separator, on every path
  in~$G'$ the unique run of~$\S$ on the sequence of priorities that
  occur on the path is accepting. 
  It follows that in the game $(G \rhd \Rr_{n, d}) \rhd \S$, using
  strategy $\sigma$ on the $G \rhd \Rr_{n, d}$ component guarantees
  that Even always wins, and hence Even has a winning strategy in the
  equivalent game $G \rhd (\Rr_{n, d} \rhd \S)$. 

  On the other hand, every word in~$\limsupodd_d$ is rejected
  by~$\Rr_{n, d}$ because it is a strong $(n, d)$-separator.  
  It follows that in every run of~$\Rr_{n, d}$ on such a word, the
  highest transition priority that occurs infinitely often is odd. 
  In other words, in every run of the chained-product automaton  
  $\Rr_{n, d} \rhd \S$ on such a word, the word over the
  alphabet~$\Sigma_{d'}$ that is fed into the automaton~$\S$ is
  in~$\limsupodd_{d'}$, and hence it is rejected because~$\S$ is a
  strong $(n', d')$-separator. 
\end{proof}

Consider taking~$\S$ to be the deterministic safety 
automaton~$\U_{n', d'}$ from Theorem~\ref{thm:sep-from-ut}.
By Proposition~\ref{proposition:parity-to-safety-separators}, the
chained-product automaton $\Rr_{n, d} \rhd \U_{n', d'}$ is a safety
good-for-separation $(n, d)$-separator. 
Moreover, if $\U_{n', d'}$ is based on the nearly optimal
quasi-polynomial universal trees from
Theorem~\ref{thm:small-universal-tree} then the size of 
$\Rr_{n, d} \rhd \U_{n', d'}$ is $n^{\lg n (\lg d + O(1))}$.

\section{Universal trees inside separating automata}
\label{sec:proof}

The main result established in this section, and a main technical
contribution of the paper, is the following quasi-polynomial lower 
bound on the size of all safety strong separators. 

\begin{theorem}
\label{thm:lower-bound}
  Every non-deterministic safety strong $(n, d)$-separator 
  has at least 
  $\binom{\floor{\lg n} + d/2 - 1}{\floor{\lg n}}$ states, 
  which is at least $n^{\lg(d/{\lg n}) - 2}$.
\end{theorem}
Since the currently known quasipolynomial time algorithms for  
solving parity games implicitly or explicitly construct a separating
automaton~\cite{CJKLS17,JL17,FJSSW17,Leh18}, the size of which
dictates the complexity, this result explains in a unified way the
quasipolynomial barrier. 
Moreover, our proof of the theorem pinpoints universal trees as the
underlying combinatorial structure behind all the recent algorithms,
by establishing that there is an $(n, d/2)$-universal tree hiding in
the set of states of every safety strong $(n, d)$-separator.

\begin{lemma}
\label{lem:tree-embedding}
  If $\A$ is a (non-deterministic) safety strong $(n, d)$-separator
  then there is an injective mapping from the leaves of some
  $(n, d/2)$-universal tree into the states of~$\A$. 
\end{lemma}
Note that Theorem~\ref{thm:lower-bound} follows from
Lemma~\ref{lem:tree-embedding} by applying
Therem~\ref{thm:lower_bound_universal_tree}---our quasi-polynomial
lower bound for universal trees, because $d \leq n$.

We prove \cref{lem:tree-embedding} in two steps. 
In the first step, we show that 
every safety strong $(n, d)$-separator 
has a \emph{tree-like} structure
(\cref{lem:tree-like}). 
Then, assuming this special structure, we prove that there is a
universal tree whose leaves can be injectively mapped into the set of
states of~$\A$ (\cref{lem:universal-tree}).  

\subsection{Tree-like structure.}

A binary relation $\preceq$ on a set $X$ is called a 
\emph{linear quasi-order} if it is reflexive, transitive, and total
(i.e. such that for all $x,y\in X$ either $x\preceq y$, or $y\preceq
x$, or both). 
If $x\preceq y$ and $y\not\preceq x$, then we write $x\prec y$. 
An equivalence class of $\preceq$ is a maximal set $e\subseteq X$ such
that $x\preceq y$ and $y\preceq x$ for all $x,y\in e$. It is
well-known that the equivalence classes of $\preceq$ form a partition
of $X$ and given two equivalence classes $e$ and $e'$, there exist $x
\in e$ and $y \in e'$ such that $x \preceq y$ if and only if for all
$x \in e$ and $y \in e'$, $x \preceq y$. 
When this is the case, it is denoted by $e \preceq e'$, and $e \prec
e'$ when additionally $e' \not\preceq e$.  

Given two linear quasi-orders $\preceq_1$ and $\preceq_2$, we write
${\preceq_1} \subseteq {\preceq_2}$ if for all $x,y \in X$, 
$x \preceq_1 y$ implies $x \preceq_2 y$. 
In that case, any equivalence class of $\preceq_2$ is formed with a
partition of equivalence classes of $\preceq_1$.  
In other words, an equivalence class of $\preceq_1$ is 
included in a unique equivalence class of $\preceq_2$ and disjoint
from the other ones. 

For an automaton $\A$ over the alphabet $\Sigma_d$, a
\emph{tree-decomposition} of $\A$ is a sequence of linear quasi-orders  
${\preceq_1} \subseteq {\preceq_3} \subseteq \dots \subseteq \,
{\preceq_{d+1}}$ 
on the set of non-rejecting states of $\A$ such that: 
\begin{compactenum}[1.]
\item
  if $(s, p, s')$ is a transition in~$\A$ then $s \succeq_{2i+1} s'$
  for all~$i$ such that $p < 2i+1 \leq d+1$. 

\item
  if $(s, p, s')$ is a transition in~$\A$ and $p$ is odd then 
  $s \succ_{2i+1} s'$ 
  for all~$i$ such that $1 \leq 2i+1 \leq p$. 

\item 
  ${\preceq_{d+1}}$ has a single equivalence class, containing all 
  non-rejecting states of~$\A$. 
\end{compactenum}
In other words, reading a priority cannot cause an increase with
respect to orders with indices greater than it, and additionally 
reading an odd priority necessarily causes a decrease with 
respect to orders whose indices are smaller than or equal to this
priority.  
If there is a tree-decomposition of $\A$, we say that $\A$ is
\emph{tree-like}.   
Given a tree decomposition $D$ of $\A$, we define the \emph{$D$-tree}
of $\A$, denoted $\tree_D(\A)$, as follows 
(recall the notation for ordered trees from
Section~\ref{sec:pm-and-ut}):  
\begin{compactitem}
\item 
  nodes of the $\tree_D(\A)$ are sequences
  $\seq{e_{d-1},e_{d-3},\dots,e_p}$, where 
  $p$ is odd, $1 \leq p \leq d-1$, 
  and where every branching direction $e_i$ is an equivalence class of
  the quasi-order $\preceq_i$, such that $e_{d-1} \supseteq e_{d-3}
  \supseteq \dotsm \supseteq e_p$, 
\item 
  the order between branching directions $e_i,e'_i$ being equivalence
  classes of $\preceq_i$ is $e_i<e'_i$ when $e_i\prec_i e_i'$. 
\end{compactitem}
Notice that for a non-rejecting state $q$ of $\A$, there is a unique
sequence $e_{d-1} \supseteq e_{d-3} \supseteq \dotsm \supseteq e_1$
where for every~$i$, $e_i$ is an equivalence class of $\preceq_i$
containing~$q$. 
One can thus assign a non-rejecting state $q$ to the corresponding
leaf $\seq{e_{d-1},e_{d-3},\dots,e_1}$ in such a way that for every
odd priority $p \in \{1,3,\ldots,d-1\}$, $q \prec_p q'$ if and only if
the $p$-truncation of the leaf assigned to $q$ is smaller in the 
lexicographic order than the $p$-truncation of the leaf assigned
to~$q'$.  

\medskip

An automaton is \emph{accessible} if for every state $q$ there
exists a run from an initial state to $q$, and moreover, if $q$ is
non-rejecting, there exists such a run which does not go through any
rejecting state. 
The first of the two steps of the proof of \cref{lem:tree-embedding}
can be summarized by the following lemma. 

\begin{lemma}\label{lem:tree-like}
  Every accessible non-deterministic safety strong $(n, d)$-separator 
  is tree-like. 
\end{lemma}

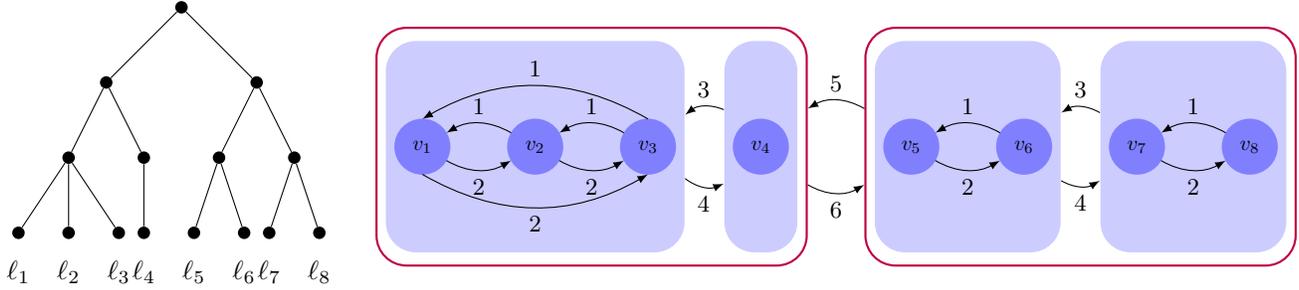
\begin{figure*}
\centering
\begin{subfigure}{0.25\textwidth}
\begin{tikzpicture}[level/.style={sibling distance = 2cm/#1,
  level distance = 1cm}]
\tikzstyle{one}=[scale=0.5, shape=circle, fill=black]
\node[one] (a) {}
child { node[one] {} 
    child { node[one] {} 
    	child { node[one] {} 
    		child [level distance=1.5em] {node {$\ell_1$} edge from parent[draw=none] } }
    	child { node[one] {} 
    		child [level distance=1.5em] {node {$\ell_2$} edge from parent[draw=none] } }
    	child { node[one] {} 
			child [level distance=1.5em] {node {$\ell_3$} edge from parent[draw=none] }
} }
    child { node[one] {} 
    	child { node[one] {} 
			child [level distance=1.5em] {node {$\ell_4$} edge from parent[draw=none] }} } }
child { node[one] {} 
	child { node[one] {} 
		child { node[one] {} 
			child [level distance=1.5em] {node {$\ell_5$} edge from parent[draw=none] } }
		child { node[one] {} 
			child [level distance=1.5em] {node {$\ell_6$} edge from parent[draw=none] }} }
	child { node[one] {} 
		child { node[one] {} 
			child [level distance=1.5em] {node {$\ell_7$} edge from parent[draw=none] } }
		child { node[one] {} 
			child [level distance=1.5em] {node {$\ell_8$} edge from parent[draw=none] } } } };
\end{tikzpicture}
\end{subfigure}
\hfill
\begin{subfigure}{0.71\textwidth}
\begin{tikzpicture}[scale=1]
\tikzset{
	every state/.style={draw=blue!50,very thick,fill=blue!50,scale=0.8},
	fleche/.style={->, >=latex}
}

\node[state] (n1) at (-6.5,0) {$v_1$};
\node[state] (n2) at (-5,0) {$v_2$};
\node[state] (n3) at (-3.5,0) {$v_3$};
\node[state] (n4) at (-2,0) {$v_4$};
\node[state] (n5) at (0,0) {$v_5$};
\node[state] (n6) at (1.5,0) {$v_6$};
\node[state] (n7) at (3,0) {$v_7$};
\node[state] (n8) at (4.5,0) {$v_8$};

\begin{pgfonlayer}{background}
\node[fill=blue!20, fit=(n1)(n2)(n3), rounded corners=4mm, minimum height=8em] (n123) {};
\node[fill=blue!20, fit=(n4), rounded corners=4mm, minimum height=8em] (n4') {};
\node[fill=blue!20, fit=(n5)(n6), rounded corners=4mm, minimum height=8em] (n56) {};
\node[fill=blue!20, fit=(n7)(n8), rounded corners=4mm, minimum height=8em] (n78) {};
\end{pgfonlayer}

\node [fit=(n123)(n4'), rounded corners=4mm, draw=purple, thick, minimum height=9em] (n1234) {};
\node [fit=(n56)(n78), rounded corners=4mm, draw=purple, thick, minimum height=9em] (n5678) {};

\path[fleche]     (n8) edge [bend right] node [above] {\small $1$} (n7);
\path[fleche]     (n7) edge [bend right] node [below] {\small $2$} (n8);

\path[fleche]     (n6) edge [bend right] node [above] {\small $1$} (n5);
\path[fleche]     (n5) edge [bend right] node [below] {\small $2$} (n6);

\path[fleche]     (n3) edge [bend right] node [above] {\small $1$} (n2);
\path[fleche]     (n2) edge [bend right] node [below] {\small $2$} (n3);

\path[fleche]     (n2) edge [bend right] node [above] {\small $1$} (n1);
\path[fleche]     (n1) edge [bend right] node [below] {\small $2$} (n2);

\path[fleche]     (n3.north) edge [bend right] node [above] {\small $1$} (n1.north);
\path[fleche]     (n1.south) edge [bend right] node [below] {\small $2$} (n3.south);

\path[fleche]     (n78.160) edge [bend right] node [above] {\small $3$} (n56.20);
\path[fleche]     (n56.340) edge [bend right] node [below] {\small $4$} (n78.200);

\path[fleche]     (n4'.135) edge [bend right] node [above] {\small $3$} (n123.12);
\path[fleche]     (n123.348) edge [bend right] node [below] {\small $4$} (n4'.225);

\path[fleche]     (n5678.170) edge [bend right] node [above] {\small $5$} (n1234.10);
\path[fleche]     (n1234.350) edge [bend right] node [below] {\small $6$} (n5678.190);

\end{tikzpicture}
\end{subfigure}
\caption{An ordered tree~$t$ and the corresponding even graph~$G_t$.}
\label{figure:tree-and-graph}
\end{figure*}

\begin{proof}
We define the linear quasi-orders inductively starting from higher
indices, that is, in the order 
${\preceq_{d+1}}, {\preceq_{d-1}}, \dots, {\preceq_1}$. 
As ${\preceq_{d+1}}$ we take the quasi-order in which all
non-rejecting states are in a single equivalence class. 
Condition 3.\ is already satisfied.

Assume now that the quasi-orders 
${\preceq_{d+1}}, {\preceq_{d-1}}, \dots, {\preceq_{2i+3}}$ are
already defined, and that they fulfil Conditions~1.\ and~2. 
We define ${\preceq_{2i+1}}$ as a refinement of ${\preceq_{2i+3}}$.    
If $q \prec_{2i+3} r$ then we set $q \prec_{2i+1} r$ as well. 
So it remains to define whether ${\preceq_{2i+1}}$ holds for
states~$q$ and~$r$ in the same equivalence class of
${\preceq_{2i+3}}$.    
Before this, notice that we are already guaranteed that
${\preceq_{2i+1}}$ satisfies Conditions 1.\ and 2.\ for priorities
higher than~$2i+1$.  
Indeed, Condition 1.\ talks only about priorities smaller than~$2i+1$. 
By Condition~2.\ applied to ${\preceq_{2i+3}}$ we know that if
priority~$p$ is odd, $p \geq 2i+3$ (i.e., $p > 2i+1$),
and $(s, p, s')$ is a transition, then $s \succ_{2i+3} s'$, so also   
$s \succ_{2i+1} s'$. 
Therefore we only need to define ${\preceq_{2i+1}}$ in such a way that 
Condition~1.\ is satisfied for priorities  
from the set $\{1, \dots, 2i\}$ and Condition~2.\ is satisfied for 
the priority~$2i+1$. 
Intuitively speaking, now we only have to care about priorities that
are at most~$2i+1$.  

Let $e$ be an arbitrary equivalence class of ${\preceq_{2i+3}}$. 
Consider an automaton $\A_{e,2i+1}$, which contains a part of $\A$
consisting of only the states that belong to the class $e$, 
and only the transitions that have endpoints inside $e$ and are
labelled by priorities at most $2i+1$ 
(notice that $\A_{e, 2i+1}$ may not be complete). 
Observe now that there cannot be any odd cycle in $\A_{e, 2i+1}$.
Indeed, otherwise we could consider the infinite run that first
reaches this cycle from an initial state in~$\A$ without visiting
any rejecting state 
(which is possible under the assumption that $\A$ is accessible), 
and then goes around this cycle forever 
(note that none of the states in the cycle is rejecting);   
the word read by this run would be in $\limsupodd_d$, but it would
be accepted by the automaton (and it should not). 
Notice that this is exactly the point in the proof of
\cref{lem:tree-like} where we need the assumption about rejecting all
the words from $\limsupodd_d$. 
Rejecting all the words from $\oddcycles_{n, d}$ would not be
sufficient, as the word described above may 
not arise from an odd graph of size~$n$, and
thus does not have to belong to $\oddcycles_{n, d}$. 
Therefore, on every path in $\A_{e, 2i+1}$, at most $|\A_{e,2i+1}|-1$ 
edges are labelled by letters with priority $2i+1$. 
Let the \emph{resistance} of a state in $\A_{e, 2i+1}$ be the maximal
number of edges labelled by priority $2i+1$ over all
paths starting in that state. 
By the previous observation, the resistance of a state is always
finite.  
Having defined the resistance, for two states~$s$ and~$s'$ in
$\A_{e, 2i+1}$, we say that $s \preceq_{2i+1} s'$ if the resistance of 
$s$ is not greater than the resistance of~$s'$. 

We have to show that such a definition of ${\preceq_{2i+1}}$ indeed 
fulfils Conditions 1.\ and 2. 
For Condition 1.\, we have to check that letters with priority smaller
than $2i+1$ never cause an increase in the quasi-order
${\preceq_{2i+1}}$. 
Consider priority~$p$, such that $p < 2i+1$ and $(s, p, s')$ is a 
transition. 
We know by the induction assumption that $s \succeq_{2i+3} s'$. 
If $s \succ_{2i+3} s'$ then as well $s \succ_{2i+1} s'$, and we are
done with Condition~1. 
Otherwise, $s$ and $s'$ are in the same equivalence class of
${\preceq_{2i+3}}$, and it is enough to show that $s$ has resistance
not greater than that of~$s'$. 
However, if $s'$ has resistance $k$ and $(s, p, s')$ is a transition,
then $s$ also has resistance at least~$k$.
Thus, state $s$ cannot have smaller resistance than $s'$.  
This implies that indeed $s \succeq_{2i+1} s'$, and 
Condition~1.\ is fulfilled.   
For Condition~2., we need to show that reading priority $2i+1$ 
causes a decrease with respect to ${\preceq_{2i+1}}$. 
Assume that there is a transition $(s, 2i+1, s')$. 
By the inductive hypothesis for ${\prec_{2i+3}}$,
Condition~1.\ implies that $s \succeq_{2i+3} s'$. 
If $s \succ_{2i+3} s'$ then we are done as before, so assume that~$s$
and~$s'$ are in the same equivalence class of~${\preceq_{2i+3}}$. 
In order to show that $s \succ_{2i+1} s'$ it is enough to observe that
$s$ has greater resistance than~$s'$, which establishes Condition~2. 
\end{proof}

\subsection{Universal tree.}

The main goal of this section is to complete the proof of
\cref{lem:tree-embedding} by showing that the tree-like structure of
every safety strong $(n, d)$-separator uncovered in the previous
section has hidden inside it an $(n, d/2)$-universal tree, whose
leaves can be injectively mapped into states of the separator.  

Let $\A$ be a safety strong $(n, d)$-separator. 
First note that we can assume that~$\A$ is accessible, by removing
the states that are not reachable from initial states
and by making rejecting the non-rejecting states which are only
reachable by visiting a rejecting state.
By \cref{lem:tree-like} we know that there is a tree-decomposition~$D$ 
of~$\A$. 
We are going to prove that $\tree_D(\A)$ is a $(n,d/2)$-universal tree 
and that there is an injective mapping from its leaves into the states
of~$\A$. 

By definition, there is a bijection between the leaves of
$\tree_D(\A)$ and the equivalence classes of ${\preceq_1}$, mapping a 
leaf $\seq{e_{d-1},e_{d-3},\dots,e_1}$ to $e_1$. 
For every equivalence class $e$ of ${\preceq_1}$, pick a state $q_e$
in~$e$.  
Consider the function mapping a leaf $\seq{e_{d-1},e_{d-3},\dots,e_1}$  
of $\tree_D(\A)$ to the state $q_{e_1}$. 
This function is injective from the leaves of $\tree_D(\A)$ into the
states of the automaton.  
Thus, in order to complete the proof of \cref{lem:tree-embedding}, it
suffices to prove the following lemma.

\begin{lemma}
\label{lem:universal-tree}
  If $\A$ is a safety strong $(n, d)$-separator then for every
  tree-decomposition~$D$ of~$\A$, the $D$-tree of~$\A$ is 
  $(n, d/2)$-universal. 
\end{lemma}

\begin{proof}
  It is enough to show that every tree~$t$ of height at most $d/2$ and 
  with at most $n$ leaves can be isomorphically embedded
  into~$\tree_D(\A)$.   
  Without loss of generality, we assume that $t$ has exactly $n$
  leaves, and that all the leaves have depth~$d/2$. 
  Otherwise, if the number of leaves is less than~$n$, add some
  branches to the tree so as to have exactly $n$ leaves; and if a leaf
  is not at depth~$d/2$, add a path from the leaf so as to reach
  depth $d/2$. 
  If such a tree can be isomorphically embedded in $\tree_D(\A)$ then
  so can the original one.  

  As a running example, consider the ordered tree $t$ with 
  $n=8$ and $d=6$ in Figure~\ref{figure:tree-and-graph}. 

  The proof proceeds by the following three steps. 
  \begin{compactenum}
  \item 
    The construction of an even graph $G_t$ with $n$ vertices and~$d$
    priorities, whose structure reflects the structure of the ordered
    tree~$t$ and, in particular, whose vertices correspond to leaves
    in~$t$. 

  \item 
    The construction of an infinite word $\alpha_t$ that is the priority  
    projection of an infinite path in~$G_t$ 
    (and that hence is in~$\evencycles_{n, d}$), 
    whose sufficiently long finite prefix is highly repetitive. 

  \item 
    A proof---using the highly repetitive nature of word
    $\alpha_t$---that $t$ can be isomorophically embedded
    into~$\tree_D(\A)$.   
\end{compactenum}

\subsubsection*{Construction of $G_t$.}
  As the set of vertices of $G_t$ we take $\{v_1, v_2, \dots, v_n\}$. 
  There is a vertex~$v_i$ for each leaf~$\ell_i$ in tree~$t$, for
  every  $i=1, 2, \dots, n$, where $\ell_1, \ell_2, \dots, \ell_n$ are
  all the leaves of tree~$t$ listed in the (increasing) lexicographic
  order. 
  Consider~$i$ and~$j$, such that $1 \leq i < j \leq n$, and let~$p$
  be the smallest priority in $\eset{1, 2, \dots, d}$, such that 
  $\ell_i|_p = \ell_j|_p$. 
  Note that then $\ell_i|_p = \ell_j|_p$ is the maximum-depth common
  ancestor of leaves~$\ell_i$ and~$\ell_j$, and~$p$ is even. 
  In graph~$G_t$, there is then an edge from~$v_i$ to~$v_j$ with 
  the even priority~$p$, and there is an edge from~$v_j$ to~$v_i$ with 
  the odd priority~$p-1$. 
  Moreover, for every vertex~$v_i$ and every even priority~$p$, 
  $1 < p \leq d$, there is a self-loop from~$v_i$ to itself in~$G_t$ 
  with the even priority~$p$. 
  Note that every cycle in graph~$G_t$ is even. 

The right-hand-side of Figure~\ref{figure:tree-and-graph} illustrates
the graph~$G_t$ for our running example tree~$t$ illustrated in the
left-hand-side of the figure. 
In order to avoid clutter, a single edge drawn in
Figure~\ref{figure:tree-and-graph} from one set of vertices
(enclosed in a rounded rectangle) to another denotes the set of edges
from every vertex in the former set to every vertex in the latter; 
moreover, the self-loops are not shown.

\begin{figure*}
\begin{center}
\begin{tikzpicture}[level/.style={sibling distance = 7cm/#1,
  level distance = 1.5cm}]
\tikzstyle{one}=[scale=0.8]
\node[one] (a) {$\alpha_x = 
  \big((1 \cdot 2)^{|\A|} \cdot 3 \cdot (1 \cdot 2)^{|\A|} \cdot 4\big)^{|\A|} \cdot 5 \cdot 
  \big((2)^{|\A|} \cdot 3 \cdot (1 \cdot 1 \cdot 2)^{|\A|} \cdot 4\big)^{|\A|}$} 
child { node[one, text width=5cm] {$\beta_x = \big((2)^{|\A|} \cdot 3
    \cdot (1 \cdot 1 \cdot 2)^{|\A|} \cdot 4\big)^{|\A|}$ \\ $\alpha_x = (2)^{|\A|} \cdot 3 \cdot (1 \cdot 1 \cdot 2)^{|\A|}$} 
    child { node[one, text width=3cm] {$\beta_x = (1 \cdot 1 \cdot
        2)^{|\A|}$ \\ $\alpha_x = 1 \cdot 1$} 
    	child { node[one] {} 
    		child [level distance=1em] {node {$\ell_1$} edge from parent[draw=none] } } 
    	child { node[one] {} 
    		child [level distance=1em] {node {$\ell_2$} edge from parent[draw=none] }}
    	child { node[one] {} 
    		child [level distance=1em] {node {$\ell_3$} edge from parent[draw=none] }} }
    child { node[one, text width=2cm] {$\beta_x = (2)^{|\A|}$ \\ $\alpha_x = \eps$} 
    	child { node[one] {} 
    		child [level distance=1em] {node {$\ell_4$} edge from parent[draw=none] }} } }
child { node[one, text width=5cm] {$\beta_x = \big((1 \cdot 2)^{|\A|}
    \cdot 3 \cdot (1 \cdot 2)^{|\A|} \cdot 4\big)^{|\A|}$ \\ $\alpha_x = (1 \cdot 2)^{|\A|} \cdot 3 \cdot (1 \cdot 2)^{|\A|}$} 
	child { node[one, text width=2cm] {$\beta_x = (1 \cdot
            2)^{|\A|}$ \\ $\alpha_x = 1$} 
		child { node[one] {} 
			child [level distance=1em] {node {$\ell_5$} edge from parent[draw=none] } }
		child { node[one] {} 
    		child [level distance=1em] {node {$\ell_6$} edge from parent[draw=none] }} }
	child { node[one, text width=2cm] {$\beta_x = (1 \cdot
            2)^{|\A|}$ \\ $\alpha_x = 1$} 
		child { node[one] {} 
    		child [level distance=1em] {node {$\ell_7$} edge from parent[draw=none] }}
		child { node[one] {} 
			child [level distance=1em] {node {$\ell_8$} edge from parent[draw=none] } } } };
\end{tikzpicture}
\end{center}
\caption{Words $\alpha_x$ and $\beta_x$ for every non-leaf node~$x$
  in the ordered tree~$t$ from Figure~\ref{figure:tree-and-graph}.}  
\label{figure:alpha-beta}
\end{figure*}
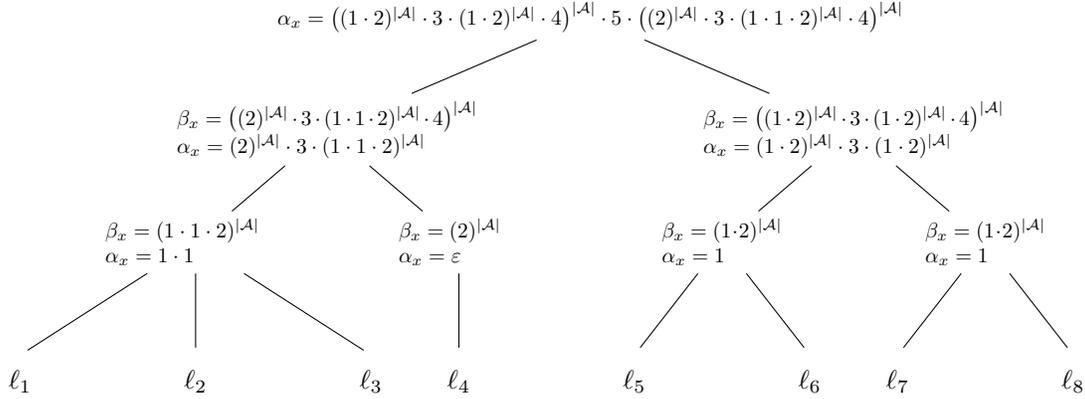

\subsubsection*{Construction of $\alpha_t$.}
We now describe an infinite path in~$G_t$ whose priority projection
gives an infinite word~$\alpha_t$;
note that the latter is in~$\evencycles_{n, d}$ because~$G_t$ is an
even graph. 
Before we give the rigorous and general construction,  we illustrate
the construction on our running example.  
In the case of graph~$G_t$ in Figure~\ref{figure:tree-and-graph}, the
infinite path is as follows. 
\begin{compactenum}
\item \label{s1} Follow the cycle between vertices $v_8$ and $v_7$ a
  large number of times, alternating between edge priorities~$1$ 
  and~$2$.  
\item \label{s2} Reach vertex $v_6$ from vertex $v_7$ with edge
  priority~$3$. 
\item \label{s3} Follow the cycle between vertices~$v_6$ and~$v_5$ a
  large number of times, alternating between edge priorities~$1$ 
  and~$2$. 
\item \label{s4} Reach vertex~$v_8$ from vertex~$v_5$ with edge
  priority~$4$. 
\item \label{s5} Repeat steps \ref{s1}, \ref{s2}, \ref{s3} and
  \ref{s4} a large number of times. 
\item \label{s6} Reach vertex~$v_4$ from vertex~$v_5$ with edge
  priority~$5$. 
\item \label{s7} Take the self-loop around vertex $v_4$, with edge 
  priority $2$, a large number of times. 
\item \label{s8} Reach vertex~$v_3$ from vertex~$v_4$ with edge
  priority~$3$. 
\item \label{s9} Follow the cycle going through vertices $v_3$, $v_2$,
  and $v_1$ a large number of times, alternating between edge
  priorities~$1$, $1$, and~$2$. 
\item \label{s10} Reach vertex $v_4$ from vertex $v_1$ with edge
  priority~$4$. 
\item \label{s11} Repeat steps~\ref{s7}., \ref{s8}., \ref{s9}., and
  \ref{s10}.\ a large number of times. 
\item After all those steps, in order to make the path and the word
  infinite, once vertex~$v_1$ is reached after step \ref{s9}., take
  the self-loop around $v_1$, with edge priority~$2$, infinitely many
  times.   
\end{compactenum}

Now we give a general and rigorous definition of an infinite
word~$\alpha_t$ for an arbitrary ordered tree~$t$ and the
corresponding even graph~$G_t$.  
For every node $x$ of $t$, we define two finite words~$\alpha_x$
and~$\beta_x$ inductively. 
If $x$ is a leaf then we set $\alpha_x = \beta_x = \eps$.   
Let $x$ be a node of depth~$d/2-k$, $0 < k \leq d/2$, and
let $x_1, x_2, \dots, x_\ell$ be the children of $x$, listed in the
increasing order. 
We define~$\alpha_x$ and~$\beta_x$ in the following way:
\[
  \alpha_x \; = \; \beta_{x_\ell} \cdot (2k-1) \cdot
  \beta_{x_{\ell-1}} \cdot (2k-1) \cdots (2k-1) \cdot \beta_{x_1}\,,
\]
and
\[
  \beta_x \; = \; \big( \alpha_x \cdot (2k) \big)^{|\A|}\,.
\]
where $|\A|$ is the number of states in $\A$. 
Finally, we set $\alpha_t = \alpha_r \cdot (2)^\omega$ where~$r$ is the
root of the tree~$t$. 
Notice that the highest priority appearing in $\alpha_t$ is at
most~$d$ (because the height of~$t$ is~$d/2$), and so 
$\alpha_t \in \Sigma_d^\omega$.

In Figure~\ref{figure:alpha-beta}, we give the pairs 
$(\alpha_x, \beta_x)$ for all non-trivial nodes~$x$ in the tree~$t$
from Figure~\ref{figure:tree-and-graph}.

\begin{proposition}
  There is an infinite path in graph~$G_t$ whose priority projection
  is the infinite word~$\alpha_t$. 
\end{proposition}

\begin{proof}
  We argue---by induction on the structure of the ordered
  tree~$t$---that for every node~$x$ in~$t$, there are two finite
  paths containing only 
  vertices in~$G_t$ that correspond to leaves in the subtree of~$t$
  rooted at~$x$, and whose priority projections are the
  words~$\alpha_x$ and~$\beta_x$, respectively.  
  The base case, when~$x$ is a leaf, is straightforward.
  The inductive case for
  \[
  \alpha_x \; = \; \beta_{x_\ell} \cdot (2k-1) \cdot
    \beta_{x_{\ell-1}} \cdot (2k-1) \cdots (2k-1) \cdot \beta_{x_1}
  \]
  when~$x$ is a node with children~$x_1, x_2, \dots, x_\ell$ follows
  routinely from the inductive hypothesis and the definition of the
  set of edges with priorities in the graph~$G_t$. 
  The claim for~$\beta_x \; = \; \big(\alpha_x \cdot (2k)\big)^{|\A|}$
  can be equally routinely obtained from the result for~$\alpha_x$. 
  Finally, extending the finite path whose priority projection
  is~$\alpha_r$, where $r$ is the root of tree~$t$, to an infinite
  path whose priority projection is~$\alpha_t$ is straightforward,
  because every vertex in~$G_t$ has a self-loop with priority~$2$. 
\end{proof}

\subsubsection*{Embedding $t$ isomorphically in~$\tree_D(\A)$.}  
We prove the following claim for every node~$x$ in~$t$;
let the depth of~$x$ be~$d/2-k$. 

\begin{claim}
  Let $\rho$ be a finite run of~$\A$ not visiting rejecting states
  ($\rho$ needs not start in an initial state) that either: 
  \begin{compactitem}
  \item	reads $\alpha_x$ and is such that all states visited by $\rho$ 
    belong to the same equivalence class of ${\preceq_{2k+1}}$, or  
  \item	reads $\beta_x$.
  \end{compactitem}
  Then there exists a node $x' = \seq{e_{d-1},e_{d-3},\dots,e_{2k+1}}$ 
  of $\tree_D(\A)$  such that 
  \begin{compactenum}[1.]
  \item the classes $e_{d-1},e_{d-3},\dots,e_{2k+1}$ contain some
    state visited by $\rho$, 
  \item the subtree of $t$ rooted at $x$ embeds isomorphically in the
    subtree of $\tree_D(\A)$ rooted at $x'$. 
  \end{compactenum}
\end{claim}

Let us first show how this claim finishes the proof of
\cref{lem:universal-tree}. 
Because $\alpha_t \in \evencycles_{n,d}$, there is
a run of $\A$ that reads $\alpha_t$ and never visits rejecting states; 
let $\rho$ be the prefix of this run that reads $\alpha_r$, where $r$ 
is the root of~$t$.
Recall that the depth of all the leaves in~$t$ is~$d/2$.
All states visited by $\rho$ belong to the same equivalence class of
${\preceq_{d+1}}$, because this quasi-order has only one equivalence 
class. 
Using the claim for the run $\rho$, we obtain that tree~$t$ embeds
isomorphically in $\tree_D(\A)$ 
(note that $x'$ is necessarily the root of $\tree_D(\A)$). 

\medskip

We now prove the claim by induction on $k$.
In the base case, when $x$ is a leaf, 
we have $\alpha_x = \beta_x = \eps$, and as branching directions of
the node $x' = \seq{e_{d-1}, e_{d-3}, \dots, e_1}$ we take the
equivalence classes of 
${\preceq_{d-1}},{\preceq_{d-3}}, \dots, {\preceq_1}$ 
containing the only state of $\rho$; 
such a node indeed fulfils Conditions 1.\ and~2. 

We now focuse on the inductive step.
Suppose first that $\rho$ reads $\alpha_x$, and that all visited
states belong to the same equivalence class of ${\preceq_{2k+1}}$. 
Let $x_1, x_2, \dots, x_\ell$ be the children of~$x$.
Recall that
\[
\alpha_x \; = \; \beta_{x_\ell} \cdot (2k-1) \cdot
  \beta_{x_{\ell-1}} \cdot (2k-1) \cdots (2k-1) \cdot \beta_{x_1}
\]
We can divide $\rho$ into fragments 
$\rho_\ell, \rho_{\ell-1}, \dots, \rho_1$, where $\rho_i$
reads~$\beta_{x_i}$.   
For every $i \in \{1, 2, \dots,\ell\}$, by the inductive hypothesis,
we can find a node $x'_i = \seq{e_{i,d-1}, e_{i,d-3}, \dots, e_{i,2k-1}}$
that fulfils Conditions~1. and~2.\ with respect to $\rho_i$ and~$x_i$.  
By assumption, all states visited by $\rho$ belong to the same
equivalence class of ${\preceq_{2k+1}}$, hence also to the same 
equivalence class of ${\preceq_j}$ for all
$j \in \{2k+1, 2k+3, \dots, d-1\}$. 
On the other hand, by Condition 2., the classes $e_{i,j}$ for 
$i \in \{1, 2, \dots,\ell\}$ and $j \in \{2k+1, 2k+3, \dots, d-1\}$
contain some state visited by~$\rho$. 
It follows that $e_{i,j} = e_{i',j}$ for all 
$i, i' \in \{1, 2, \dots, \ell\}$ and 
$j \in \{2k+1, 2k+3, \dots, d-1\}$, 
which implies that the nodes $x_1', x_2', \dots, x_\ell'$ are
siblings, having a common parent~$x'$.  
It is easy to see that the node~$x'$ satisfies Conditions 1.\ and~2.

Suppose now that $\rho$ reads 
$\beta_x = \big(\alpha_x \cdot (2k)\big)^{|\A|}$.  
By the fact that $D$ is a tree-decomposition of $\A$,  
we know that no transition in $\rho$ goes up with respect to the
quasi-order ${\preceq_{2k+1}}$, because all priorities in~$\beta_x$
are at most~$2k$. 
As $\A$ has at most $|\A|$ states, it means that at most $|\A|-1$
transitions of the considered run cause a decrease with respect
to~${\preceq_{2k+1}}$.  
This implies that there is a part of this run that reads $\alpha_x$
and contains no increase nor decrease with respect to
${\preceq_{2k+1}}$, that is, all states visited by this part of the
run belong to the same equivalence class of ${\preceq_{2k+1}}$.   
We continue with this part of the run as in the previous paragraph,
and in this way we finish the proof of the claim. 
\end{proof}

\section{Open questions and further work}

We leave open the question whether a stronger version of our main
technical result holds, namely whether every safety automaton
separating $\evencycles_n$ from $\oddcycles_n$ has at least
quasi-polynomial number of states.
Our argument cannot be directly extended to that setting, as already
the proof of \cref{lem:tree-like} crucially relies on that fact that
no word from $\limsupodd$ is accepted. 

We also leave it as open questions whether there are parity automata 
as in Theorem~\ref{theorem:lehtinen} but with a sub-quasi-polynomial
number of states and a logarithmic number of priorities, or whether
quasi-polynomial lower bounds can also be established for such parity
automata. 

A very recent technical report by Colcombet and Fijalkow~\cite{CF18}
gives alternative proofs of our Theorem~\ref{thm:sep-from-ut} and 
Lemma~\ref{lem:universal-tree}. 
They do so by introducing a new concept of \emph{universal graphs},
and they argue that a subclass of universal graphs that satisfies an 
extra (``saturation'') property corresponds to universal trees. 
Their main result is that the sizes of the smallest strong separators,
universal trees, and universal graphs are equal to one another, which
is a refinement of the corollary of our Theorem~\ref{thm:sep-from-ut}
and Lemma~\ref{lem:universal-tree} that the former two are.
However, in order to conclude that all those smallest sizes are
quasi-polynomial, they too rely on our
Theorem~\ref{thm:lower_bound_universal_tree}.

\section*{Acknowledgements}

This research has been supported by the EPSRC grant EP/P020992/1
(Solving Parity Games in Theory and Practice).
W.~Czerwi\'nski and P.~Parys are partially supported by the Polish
National Science Centre grant 2016/21/D/ST6/01376.  
N. Fijalkow is supported by The Alan Turing Institute under the EPSRC
grant EP/N510129/1 and the DeLTA project (ANR-16-CE40-0007).

\bibliographystyle{plain}
\bibliography{citat}

\end{document}